\documentclass[12pt, reqno, article]{amsart}
\usepackage{setspace}
\usepackage[foot]{amsaddr}
\usepackage{bm}
\usepackage{tikz-qtree}
\usepackage{booktabs}
\usepackage{multirow}
\usepackage{bbm}
\usepackage{diagbox}
\usepackage{algorithm, algpseudocode}
\usetikzlibrary{positioning}
\usepackage[per=slash]{siunitx}

\newtheorem{proposition}{Proposition}

\newtheorem{lemma}{Lemma}

\theoremstyle{definition}
\newtheorem{condition}{Condition}

\theoremstyle{remark}

\newtheorem{remark}{Remark}
\newcommand{\var}{\text{var}}
\newcommand{\se}{\text{se}}
\newcommand{\sd}{\text{sd}}
\newcommand{\tr}{\text{tr}}
\newcommand{\limn}{\underset{n \to \infty}{\lim}}

\newcommand{\pr}{\text{Pr}}
\newcommand{\E}{{E}}
\renewcommand{\P}{{\text{pr}}}

\newcommand{\bH}{{\mathbf{H}}}

\newcommand{\bx}{{\mathbf{x}}}
\newcommand{\bZ}{{\mathbf{Z}}}
\newcommand{\bQ}{{\mathbf{Q}}}
\newcommand{\bB}{{\mathbf{B}}}
\newcommand{\bHQ}{{\mathbf{H_Q}}}
\newcommand{\bI}{{\mathbf{I}}}
\newcommand{\bz}{{\mathbf{z}}}

\newcommand{\bY}{{\mathbf{Y}}}
\newcommand{\bX}{{\mathbf{X}}}

\newcommand{\1}{\mathbbm{1}}

\newcommand{\cF}{{\mathcal{F}}}

\newcommand{\cD}{{\mathcal{D}}}
\newcommand{\cC}{{\mathcal{C}}}

\newcommand{\cZ}{{\mathcal{Z}}}

\usepackage{natbib}
\bibpunct{(}{)}{;}{a}{}{,}

\allowdisplaybreaks

\title[]{Biased Encouragements and Heterogeneous Effects in an Instrumental Variable Study of Emergency General Surgical Outcomes}\thanks{The dataset used for this study was purchased with a grant from the Society of American Gastrointestinal and Endoscopic Surgeons. Although the AMA Physician Masterfile data is the source of the raw physician data, the tables and tabulations were prepared by the authors and do not reflect the work of the AMA. The Pennsylvania Health Cost Containment Council (PHC4) is an independent state agency responsible for addressing the problems of escalating health costs, ensuring the quality of health care, and increasing access to health care for all citizens. While PHC4 has provided data for this study, PHC4 specifically disclaims responsibility for any analyses, interpretations or conclusions. Some of the data used to produce this publication was purchased from or provided by the New York State Department of Health (NYSDOH) Statewide Planning and Research Cooperative System (SPARCS). However, the conclusions derived, and views expressed herein are those of the author(s) and do not reflect the conclusions or views of NYSDOH. NYSDOH, its employees, officers, and agents make no representation, warranty or guarantee as to the accuracy, completeness, currency, or suitability of the information provided here. The authors declare no conflicts.}
\author{Colin B. Fogarty$^{1}$} \address{$^{1}$ Operations Research and Statistics Group, MIT Sloan School of Management, Massachusetts Institute of Technology, Cambridge, MA } 
\author{ Kwonsang Lee$^{2}$}\address{ $^{2}$ Department of Statistics, Sungkyunkwan University, Seoul, Republic of Korea}
 \author{Rachel R. Kelz$^{3}$}\address{$^{3}$Center for Surgery and Health Economics, Department of Surgery, Hospital of the University of Pennsylvania, Philadelphia, PA} 
 \author{Luke J. Keele$^{4}$}\address{$^4$ University of Pennsylvania, 3400 Spruce St, Philadelphia, PA}

\begin{document}

\begin{abstract}
We investigate the efficacy of surgical versus non-surgical management for two gastrointestinal conditions, colitis and diverticulitis, using observational data. We deploy an instrumental variable design with surgeons' tendencies to operate as an instrument. Assuming instrument validity, we find that non-surgical alternatives can reduce both hospital length of stay and the risk of complications, with estimated effects larger for septic patients than for non-septic patients. The validity of our instrument is plausible but not ironclad, necessitating a sensitivity analysis. Existing sensitivity analyses for IV designs assume effect homogeneity, unlikely to hold here because of patient-specific physiology. We develop a new sensitivity analysis that accommodates arbitrary effect heterogeneity and exploits components explainable by observed features. We find that the results for non-septic patients prove more robust to hidden bias despite having smaller estimated effects. For non-septic patients, two individuals with identical observed characteristics would have to differ in their odds of assignment to a high tendency to operate surgeon by a factor of 2.34 to overturn our finding of a benefit for non-surgical management in reducing length of stay. For septic patients, this value is only 1.64. Simulations illustrate that this phenomenon may be explained by differences in within-group heterogeneity.
 
\end{abstract}




\maketitle
\thispagestyle{empty}

\doublespace
\newpage
\setcounter{page}{1}
\section{Introduction}

\subsection{Emergency general surgery conditions and operative care}
\label{sec:motiv_example}

Emergency general surgery (EGS) refers to medical emergencies where the injury is most often endogenous (a burst appendix) while trauma care refers to injuries that are exogenous (a gunshot wound). There are a set of medical conditions known as EGS conditions that are thus distinct from trauma injuries \citep{gale2014public,shafi2013emergency}. While operative management is often the primary course of treatment for EGS conditions, it carries additional risk of complications and adverse reactions to anesthesia. For many EGS conditions, non-operative alternatives including observation, minimally invasive procedures and supportive care exist. Therefore, a critical clinical question faced by the acute care community is determining the overall efficacy of surgical interventions for EGS conditions relative to alternative courses of treatment. 

We focus on the effectiveness of operative care for two common EGS conditions: ulcerative colitis and diverticulitis. Diverticula are marble-sized pouches that can form in the lining of the lower part of the large intestine. Diverticulitis occurs when one or more diverticula tear and become inflamed or infected, causing severe abdominal pain, fever, and nausea. Ulcerative colitis is a form of inflammatory bowel disease that causes long-lasting inflammation and ulcers in the lining of the large intestine. The condition most often begins gradually but can worsen, leading to life-threatening complications. Both conditions may be managed through diet and medication, but may require surgery to remove sections of the intestine.

One challenge in studying the effectiveness of surgery for EGS conditions is that randomized comparisons between surgical and non-surgical care are generally judged to violate the principle of clinical equipoise.  Evidence must instead be collected from observational studies. Unsurprisingly, patients who undergo emergency operative treatment are often the sickest patients \citep{shafi2013emergency}. While statistical adjustment for measures of patient frailty may render treated and control groups comparable on observables, such adjustments are often insufficient due to the potential presence of hidden biases \citep{keele2018icubeds}.

\subsection{Instrumental variables and physician preference as an instrument}
\label{sec:ivintro}

Instrumental variable (IV) methods are a set of statistical techniques that can facilitate the identification of causal effects in the presence of hidden bias. An IV is a variable that is associated with the treatment of interest but affects outcomes only indirectly through its impact on treatment assignment, providing a haphazard nudge towards taking a treatment \citep{Angrist:1996}. For a variable to be an instrument, (1) the IV must be associated with the exposure; (2) the IV must be randomly or as-if randomly assigned; and (3) the IV cannot have a direct effect on the outcome \citep{Angrist:1996}. Under (1)-(3), an IV provides a consistent estimate of a causal effect even in the presence of unobserved confounding between the exposure and the outcome. See \citet{Baiocchi:2014} and \citet{imbens2014instrumental} for reviews.

One frequently proposed instrument in comparative effectiveness research is a physician's preference for a specific course of treatment. Here, it is assumed that for a patient of a given level of disease severity, physicians may differ in their preferred course of treatment for primarily idiosyncratic reasons.  When a patient receives care from a physician, it is then assumed they are being exposed to an as-if randomized encouragement towards that physician's preferred treatment. \citet{brookhart2006evaluating} first used this type of IV in the context of drug prescriptions. See \citet{brookhart2007preference} for an overview. 

We use a surgeon's preference for operative management as an instrument for whether a patient receives surgery after admission to the emergency department. Following \citet{keeleegsiv2018}, we construct a measure of a surgeon's tendency to operate (TTO) by calculating the percentage of times a surgeon operates when presented with an EGS condition. For the surgeons in the data, TTO varied from less than 5\% to 95\%. Consider two patients in the emergency department with similar baseline characteristics. Being assigned to a patient with a higher or lower TTO may then be viewed as a nudge towards or away from surgery \citep{keeleegsiv2018}. Assignment of surgeons with varying TTO to patients is a plausibly haphazard process, with clear impact on whether or not a patient receives operative care. 

Our data set merges the American Medical Association (AMA) Physician Masterfile with all-payer hospital discharge claims from New York, Florida and Pennsylvania in 2012-2013. The study population consists of all patients admitted for emergency or urgent inpatient care. The data include an identifier for each patient's surgeon, patient sociodemographic and clinical characteristics including indicators for frailty, severe sepsis or septic shock, and 31 comorbidities based on Elixhauser indices \citep{elixhauser1998comorbidity}. The data also include measures of the surgeon's age and years of experience. Surgeons were excluded if they did not perform at least 5 operations for one of the 51 specific EGS conditions per year within the two-year study time-frame. This resulted in $N=44,082$ patients presenting with either colitis or diverticulitis. Our primary outcome was the presence of any in-hospital complication (infection, bleeding) as measured by a binary variable. Secondary outcomes are hospital length-of-stay; and in-hospital mortality within 30 days after admission. We also investigated whether the presence of sepsis (an inflammatory response to infection that can result in organ malfunction) acts as a modifier for the effect of surgery. 

\subsection{Probing the robustness of an instrumental variable analysis and the roles of effect heterogeneity}

For surgeons' tendency to operate to be a valid instrument one must assume that there are no unobserved variables which simultaneously influence the TTO for the surgeon assigned to a patient and the patient's health outcomes. Imagine that surgeons with a higher TTO are more skilled surgeons. If the sickest patients are treated by the most skilled surgeons, then the resulting estimate of surgery's efficacy may be biased. While the data include measures of surgeon age and experience, these are imperfect measures of surgical skill, such that a critic of the study may be unconvinced that bias has been removed. 

Suspicions of residual bias motivate sensitivity analyses for IV studies. A sensitivity analysis asks how strong the impact of a hidden bias would have to be on the instrumental variable to alter the substantive conclusions of a study. Existing methods for sensitivity analysis share a common limitation: they are derived under a pre-specified model of effects for each individual. The most prevalent model for effects is that of proportional doses - that the difference in potential outcomes under encouragement to treatment versus encouragement to control is proportional to the difference in treatment actually received under encouragement to treatment versus encouragement to control \citep{Imbens:2005b}. Under this model, all patients who would undergo surgery if and only if assigned to a surgeon with a high TTO would have the same treatment effect. 

Treatment effect heterogeneity refers to variation in the effects of an intervention across individuals. While some of this variation may be ascribed to intrinsically random variation, a portion of it may be predictable on the basis of observed covariates known as \textit{effect modifiers}. Our work synthesizes and builds upon a growing recent literature on effect heterogeneity in matched designs. \citet{Baiocchi:2010} presents a mode of inference for matched IV studies with heterogeneous effects, but only does so under an assumption of no unmeasured confounding. \citet{Fogarty:2018} presents a general recipe for variance estimation for finely stratified designs but does not address variance estimation within a sensitivity analysis. \citet{Fogarty:2019} develops a sensitivity analysis valid under effect heterogeneity, but employs a standard error estimator that can be unduly conservative and does not address the validity of sensitivity analyses using McNemar's test with binary outcomes when effects vary.

We first present a sensitivity analysis for matched IV studies that is asymptotically correct for a weaker null hypothesis on a parameter known as the effect ratio \citep{Baiocchi:2010, kang2016full}. This method obviates concerns that the nonidentified components of effect heterogeneity may conspire with hidden bias to render attempts at sensitivity analysis insufficient. In the particular case of binary outcomes, we prove that sensitivity analysis in \citet{Baiocchi:2010}, justified therein only under the sharp null, is actually a special case of our method and hence also provides an asymptotically correct sensitivity analysis even with heterogeneous effects. We next highlight how the identified aspects of effect heterogeneity can improve the performance of a sensitivity analysis. We develop improved standard error estimates for use within sensitivity analyses with heterogeneous effects. These standard errors are valid regardless of the truth of an underlying model for effect modification as a function of observed covariates, but provide more substantial gains should the posited model of effects be correct. By further comparing results between the septic and non-septic subgroups, we provide an illustration of the role that latent heterogeneity can play in the resulting power of a sensitivity analysis as originally discussed in \citet{Rosenbaum:2005b}. 

\section{A matched instrumental variable study and its experimental ideal}
\subsection{A near-far match}

Stronger instruments provide both more informative effect estimates and estimates that are more resistant to hidden bias, providing motivation to actively strengthen an instrument as part of the analysis \citep{Small:2008}. One approach for strengthening an instrument is a near-far matched design  \citep{Baiocchi:2010}. Near-far matching creates matched pairs that are similar in terms of observed covariates but highly dissimilar on the values of the instrument. We first categorize surgeons as high versus low tendancy to operate depending on whether or not their TTO value is above or below the median TTO (0.604). We then use a form of near-far matching that combines refined covariate balance with near-far matching \citep{Pimentel:2015a,keeleicumatch2019}. This method allows us to balance a large number of nominal covariates while maintaining a large distance within pairs on the numerical value for TTO.

Each matched pair is exactly matched on both hospital and an indicator of sepsis, such that across-hospital differences in quality of care and fundamental differences in physiology between septic and non-septic patient cannot bias the analysis. For the remaining covariates, we minimized the total of the within-pair distances on covariates as measured by the Mahalanobis distance and applied a caliper to the propensity score. We further applied near-fine balance constraints for three indicators for surgical volume at the hospital at the time of admission and a dummy variable for whether a patient had a specific disability. Fine balance constrains two groups to be balanced on a particular variable without restricting matching on the variable within individual pairs \citep{Rosenbaum:2007b}. A near-fine balance constraint returns a finely balanced match when one is feasible, and otherwise minimizes the deviation from fine balance \citep{Yang:2012}. As a result, the marginal distribution for these three indicators will be exactly or nearly exactly the same across levels of the IV. Finally, we applied optimal subsetting to discard matched pairs with high levels of imbalance on covariates \citep{Rosenbaum:2011}. After matching, we calculated the standardized difference for each covariate, which is the mean difference between matched patients divided by the pooled standard deviation before matching. We attempted to produce absolute standardized differences of less than 0.10, a commonly recommended threshold \citep{Rosenbaum:2010}.

Table~\ref{tab.bal} contains balance statistics before and after matching. Matching moved most absolute standardized differences below 0.1, suggesting successful attenuation of potential biases due to discrepancies in observed surgeon characteristics and available patient physiology metrics. Based upon observed covariates, Table \ref{tab.bal} suggests that surgeons with high versus low TTO appear to have been assigned to patients in an asystematic manner.

\begin{table}[ht]
\scriptsize{
\centering
\caption{Balance before and after matching. Instrument split into high vs. low categories at the median value for TTO (0.604). Before and after matching, average values for each variable are given for the high and low groups along with the standardized difference.}
\label{tab.bal}
\begin{tabular}{lcccccc}
  \toprule
 & \multicolumn{3}{c}{Unmatched Data} & \multicolumn{3}{c}{Matched Data} \\
 & \multicolumn{3}{c}{N=44,082 (11,344 Septic)} & \multicolumn{3}{c}{N=6,068 (2,336 Septic)}\\
 & High TTO & Low TTO & St. Dif. & High TTO & Low TTO & St. Dif.\\
  \midrule
  TTO & 0.83 & 0.33 & 2.82 & 0.84 & 0.27 & 4.18 \\  
  Age & 63.48 & 65.24 & -0.11 & 68.21 & 68.13 & 0.01 \\ 
  No. Comorbidities & 3.09 & 3.11 & -0.01 & 3.99 & 3.96 & 0.01 \\ 
  Surgeon Age & 54.16 & 52.95 & 0.12 & 53.02 & 51.63 & 0.15 \\ 
  Surgeon Experience (Years) & 17.07 & 15.45 & 0.16 & 15.59 & 14.70 & 0.09 \\ 
  Female & 0.52 & 0.55 & -0.07 & 0.56 & 0.56 & 0.00 \\ 
  Hispanic & 0.11 & 0.12 & -0.04 & 0.08 & 0.09 & -0.03 \\ 
  White & 0.80 & 0.73 & 0.17 & 0.76 & 0.76 & 0.00 \\ 
  African-American & 0.10 & 0.15 & -0.13 & 0.13 & 0.13 & -0.00 \\ 
  Other Racial Cat. & 0.09 & 0.12 & -0.09 & 0.11 & 0.11 & 0.00 \\ 
  Sepsis & 0.29 & 0.24 & 0.11 & 0.38 & 0.38 & 0.00 \\ 
  Disability & 0.07 & 0.08 & -0.06 & 0.10 & 0.10 & 0.00 \\
    & \multicolumn{6}{c}{Select Comorbidities} \\
  Congestive Heart Failure & 0.15 & 0.15 & -0.02 & 0.25 & 0.23 & 0.05 \\ 
  Cardiac Arrhythmias & 0.26 & 0.23 & 0.06 & 0.37 & 0.35 & 0.04 \\ 
  Valvular Disease & 0.07 & 0.07 & 0.01 & 0.12 & 0.10 & 0.07 \\ 
  Pulm. Circulation Disorders & 0.05 & 0.04 & 0.01 & 0.08 & 0.07 & 0.03 \\ 
  Peripheral Vascular Disorders & 0.08 & 0.07 & 0.06 & 0.10 & 0.11 & -0.02 \\ 
  Hypertension, Uncomplicated & 0.43 & 0.45 & -0.02 & 0.44 & 0.45 & -0.02 \\ 
  Paralysis & 0.02 & 0.01 & 0.03 & 0.03 & 0.02 & 0.01 \\ 
  Other Neurological Disorders & 0.08 & 0.08 & 0.00 & 0.13 & 0.11 & 0.06 \\ 
  Chronic Pulmonary Disease & 0.22 & 0.22 & -0.00 & 0.27 & 0.26 & 0.04 \\ 
  Diabetes, Uncomplicated & 0.17 & 0.21 & -0.09 & 0.22 & 0.23 & -0.01 \\ 
  Diabetes, Complicated & 0.05 & 0.04 & 0.02 & 0.06 & 0.07 & -0.02 \\ 
  Hypothyroidism & 0.13 & 0.13 & -0.01 & 0.15 & 0.15 & -0.00 \\ 
  Renal Failure & 0.15 & 0.17 & -0.05 & 0.24 & 0.26 & -0.04 \\ 
  Liver Disease & 0.04 & 0.06 & -0.07 & 0.05 & 0.06 & -0.06 \\ 
  Peptic Ulcer, Excl. Bleeding & 0.01 & 0.02 & -0.06 & 0.01 & 0.01 & 0.03 \\ 
  Hypertension, Complicated & 0.14 & 0.16 & -0.03 & 0.22 & 0.24 & -0.03 \\ 
   \bottomrule
\end{tabular}
}
\end{table}

\subsection{Notation for randomized encouragement designs and instrumental variable studies}

We now introduce notation for the experimental design that near-far matching emulates. To assist the reader, essential notation and definitions introduced here and elsewhere are summarized in \S A of the web-based supplement. There are $2n$ individuals partitioned into $n$ matched pairs. Each pair contains one individual encouraged to take the treatment, denoted $Z_{ij}=1$, and one individual not encouraged to take the treatment ($Z_{ij}=0$). In our study, the patient assigned to the surgeon with the higher value of the instrumental variable $TTO$ in each pair is the patient for whom $Z_{ij}=1$.  Individual $ij$ has two potential responses: one under encouragement to treatment, $y_{ij}(1)$,  and the other under no encouragement, $y_{ij}(0)$. Further, individual $ij$ has two values for the level of treatment actually received: that under encouragement, $d_{ij}(1)$, and under no encouragement, $d_{ij}(0)$. In this work we will consider binary values for $d_{ij}(z)$ for $z=0,1$, such that $d_{ij}(z)$ represents whether or not an individual would actually take the treatment when assigned encouragement level $z$. That said, the methods we develop readily extend to the case of continuous $d_{ij}(z)$. The observed outcome and exposure are $Y_{ij} = y_{ij}(Z_{ij})$ and $D_{ij} = d_{ij}(Z_{ij})$ respectively. Let $\mathcal{F} = \{y_{ij}(1), y_{ij}(0), d_{ij}(1), d_{ij}(0), \mathbf{x}_{ij}, u_{ij}: i=1, \ldots ,n; j=1, 2\}$, where $\mathbf{x}_{ij}$ represent the observed covariates for each individual and $u_{ij}$ is an unobserved covariate with domain $0 \leq u_{ij} \leq 1$. While a univariate hidden variable constrained to the unit interval may appear restrictive, in fact there always exists a hidden variable between 0 and 1 such that if the practitioner had access to it, adjustment for it would be sufficient for identifying a causal effect \citep[\S9, footnote 15]{Rosenbaum:2017}. We consider finite-population inference, which will be reflected by conditioning upon $\mathcal{F}$ in the probabilistic statements that follow. Inferential statements will pertain to parameters defined in the observed study population, without reliance upon the existence of a hypothetical superpopulation. Let $\mathbf{Z} = (Z_{i1},Z_{i2},...,Z_{n2})^T$, and let boldface similarly denote other vector quantities such as $\mathbf{u} = (u_{i1},u_{i2},...,u_{n2})^T$. 

Patients are paired using observed covariates $\mathbf{x}_{ij}$ such that $\mathbf{x}_{i1}\approx \mathbf{x}_{i2}$. Despite this, paired subjects may not be equally likely to be exposed to high versus low values of the IV due to discrepancies on an unobserved covariate, i.e. $u_{i1} \neq u_{i2}$. Let $\Omega = \{\mathbf{z}: z_{i1}+z_{i2}=1\}$ be the set of the $2^n$ possible values $\mathbf{z}$ of $\mathbf{Z}$, and let $\mathcal{Z}$ denote the event $\{\mathbf{Z}\in \Omega\}$. In finite-population causal inference the only source of randomness is the assignment of the encouragement $Z_{ij}$. Were $\mathbf{Z}$ chosen uniformly as in a randomized experiment, we would have $\text{Pr}(\mathbf{Z}=\mathbf{z}\mid \mathcal{F}, \mathcal{Z}) = 1/|\Omega|$ for each $\mathbf{z}\in \Omega$, where $|A|$ is the number of elements in $A$.

\subsection{Conventional IV Assumptions}\label{sec:IVassume}

Here we outline common assumptions made within IV designs. First, our notation implicitly assumes the stable unit treatment value assumption \citep{rub80}: both the potential treatments received $d_{ij}(z)$ and potential outcomes $y_{ij}(z)$, $z=0,1$, depend solely on the value of the instrument for individual $ij$ and are not affected by the value of $Z_{i'j'}$ for $i'j'\neq ij$. Next, we assume that the instrument has a non-zero effect on the treatment received, $n^{-1}\sum_{i=1}^n \{d_{ij}(1)-d_{ij}(0)\}\neq 0$, an assumption which may be verified with empirical tests known as weak instrument tests. IV designs with stronger instruments are more robust \citep{Small:2008,Baiocchi:2010,Keele:2014b}, which motivated our use of near-far matching to increase instrument strength.

For the instrumental variable to yield identification of a causal effect, one must assume that instrument assignment is unconfounded once we condition on baseline covariates: $\P\{Z_{ij}=1\mid x_{ij}, d_{ij}(1), d_{ij}(0), y_{ij}(1), y_{ij}(0)\} = \P(Z_{ij}=1 \mid x_{ij})$. For instruments that arise from noncompliance in randomized trials this assumption is satisfied by design, but for naturally occurring IVs this assumption is frequently subject to debate. Since we are focused on emergency admissions it is unlikely that patients themselves are selecting their surgeon, such that assignment to a surgeon with a specific TTO is a plausibly haphazard process. While we employ matching to remove possible bias from haphazard assignment, we cannot ensure that the study results are insensitive to the presence of an unobserved confounding between the instrument and the outcome. This motivates our development in \S \ref{sec:senshet} of a sensitivity analysis to assess robustness to departures from this assumption.

In defining an instrumental variable, \citet{Angrist:1996} further require that the variable satisfies the \textit{exclusion restriction}: the instrument can only affect the outcome by influencing the treatment received, stated formally as $d_{ij}(z) = d_{ij}(z') \Rightarrow y_{ij}(z) = y_{ij}(z')$. Informally, under the exclusion restriction, any effect of a surgeon's TTO on the outcome must only be a consequence of the medical effects of operative management. We bolster this assumption by comparing patients who receive care in the same hospital, which should hold fixed other systemic factors of care that might affect outcomes. Typically, IV studies further invoke a monotonicity assumption of the following form: for all individuals, $d_{ij}(z') \geq d_{ij}(z)$ for $z' \geq z$. Together the exclusion restriction and monotonicity imply that the causal estimand being estimated by an IV design is the average treatment effect among individuals who comply with their encouragement to receive or not receive the treatment, referred to as the \textit{local average treatment effect}.  For preference-based instruments the assumption of stochastic monotonicity is more plausible, which yields a different causal estimand where individuals more strongly influenced by the instrument receive larger weight \citep{small2017instrumental}. See \citet{keeleegsiv2018} for a detailed discussion of the monotonicity assumption in the context of TTO. As will be described in Section \ref{sec:effectratio}, inference may still proceed using a variable which violates monotonicity (be it deterministic or stochastic) and/or the exclusion restriction; however, the estimand will no longer be the local average treatment effect.

\subsection{A model for biased encouragements}

While assignment of surgeons of varying tendencies to operate to patients may have appeared haphazard on the basis of observed covariates, matching can do nothing to preclude hidden biases from influencing effect estimates and hypothesis tests. We employ the model for a sensitivity analysis from \citet[\S 4]{Rosenbaum:2002}. Let $\pi_{ij} = \pr(Z_{ij}=1 \mid \mathcal{F})$. The model proposes a logit form for $\pi_{ij}$ dependent upon a matched-set specific parameter $\kappa_i$ along with the unobserved covariate $u_{ij}$,
\begin{equation}\label{eqn:sensi}
	\log\left(\frac{\pi_{ij}}{1-\pi_{ij}}\right) = \kappa_i + \log(\Gamma)u_{ij},	
\end{equation}
for some $\Gamma \geq 1$. This is equivalent to assuming that for two individuals $j$ and $k$ in the same matched set $i$, $\pi_{ij}$ and $\pi_{ik}$ may differ in the odds ratio scale by at most a factor $\Gamma$, i.e. $\Gamma^{-1}\leq \pi_{ij}(1-\pi_{ik})/\{\pi_{ik}(1-\pi_{ij})\}\leq \Gamma$; see \citet[\S 4, Proposition 12]{Rosenbaum:2002} for a proof. Returning attention to the paired structure by conditioning upon $\mathcal{Z}$, \eqref{eqn:sensi} implies $1/(1+\Gamma) \leq \text{Pr}(Z_{ij}=1\mid \mathcal{F}, \mathcal{Z}) \leq \Gamma/(1+\Gamma)$ for $i=1, \ldots ,n; j=1, 2$. In our study, \eqref{eqn:sensi} suggests that patients in the same matched set could differ in their odds ratio of being assigned to a surgeon with a high TTO by at most $\Gamma$. $\Gamma=1$ corresponds to a randomized encouragement design where $u_{ij}$ does not impact the assignment probabilities \citep{Holland:1988}, while pushing $\Gamma$ beyond one allows unmeasured confounding to increasingly bias the encouragement to treatment. In terms of the conventional IV assumptions, $\Gamma > 1$ allows for a  violation of instrument unconfoundeness conditional upon baseline covariates.

\section{The effect ratio and methods for sensitivity analysis}

\subsection{The effect ratio and its interpretation}
\label{sec:effectratio}

In an IV design, one estimand of interest is the effect ratio. The effect ratio is
\begin{equation}
	\lambda = \frac{\sum_{i=1}^{n}\sum_{j=1}^{2} \{y_{ij}(1) - y_{ij}(0)\}}{\sum_{i=1}^{n}\sum_{j=1}^{2} \{d_{ij}(1) - d_{ij}(0)\}},
\end{equation} 
where it is assumed that $\sum_{i=1}^{n}\sum_{j=1}^{2} \{d_{ij}(1) - d_{ij}(0)\} \neq 0$. It is the ratio of two average treatment effects. In our application the treatment effect in the numerator is the effect of being assigned to a surgeon with high versus low TTO on any one of the health outcomes of interest, while the effect of surgeon TTO on whether or not a patient has surgery is in the denominator. Under the assumption of monotonicity and the exclusion restriction described in \S \ref{sec:IVassume}, $\lambda$ may then be interpreted as the sample average treatment effect among compliers, i.e. the individuals who would undergo surgery if and only if assigned to a surgeon with a higher preference for surgery. Assuming unconfoundedness, \citet{Baiocchi:2010} provide a large-sample method for constructing confidence intervals and performing inference for the effect ratio under effect heterogeneity. 

Table~\ref{tab.out} shows estimated effect ratios and confidence intervals using their method for septic and non-septic patients on our three surgical outcomes. The estimated length of stay among compliers is much longer--4 days for non-septic patients and nearly a week for septic patients. Additionally, we find that compliers that underwent surgery were much more likely to experience a post-operative complication. Among septic patients, the risk of a complication was 25\% higher, and the lower bound on the 95\% confidence interval indicates the risk of a complication was 19\% or higher. Finally, the effect of surgery on mortality among compliers is positive but the confidence intervals are wide and include zero.

\begin{table}[ht]
\centering
\caption{Estimated effect of surgery on EGS patient outcomes. Point estimates for binary outcomes are risk differences expressed as percentages. Brackets are 95\% confidence intervals from inverting the test of \citet{Baiocchi:2010}.}
\label{tab.out}
\begin{tabular}{llcccc}
  \toprule
& & \multicolumn{2}{c}{Septic Patients} & \multicolumn{2}{c}{Non-septic Patients} \\
  \midrule
 & {Complication}  & 25.3 & [ 19.3 , 31.4 ] & 12.0 &[ 8.3 , 15.7 ]  \\ 
&{Length of Stay}& 6.80 & [ 4.4 , 9.2 ] & 4.10 &  [ 3.4 , 4.8 ]\\
 & {Mortality}  & 1.96 & [ -1.9 , 5.9 ] & 0.59 & [ -0.5 , 1.6 ] \\ 
\bottomrule
\end{tabular}
\end{table}



\subsection{Sensitivity analysis for the effect ratio with effects proportional to dose}\label{sec:propdose}

The estimates in Table~\ref{tab.out} are valid with heterogeneous effects under the assumption of no hidden bias. \citet{Rosenbaum:1996} and \citet[\S 5.4]{Rosenbaum:2002} develop methods for exact sensitivity analyses in paired instrumental variable studies under the proportional dose model, which states
 $$
 H_{prop}^{(\lambda_0)} : y_{ij}(1) - y_{ij}(0) = \lambda_0 \{d_{ij}(1) - d_{ij}(0)\}, i=1, \ldots ,n; j=1, 2.
 $$
This extends the model of constant effects common in randomized experiments to randomized encouragement designs. If an individual complies with their encouraged treatment then $y_{ij}(1)-y_{ij}(0)=\lambda_0$, while if an individual defies their encouragement then $y_{ij}(1)-y_{ij}(0) = -\lambda_0$ as $d_{ij}(1)-d_{ij}(0) = -1$. If the encouragement does not influence the treatment received, $y_{ij}(1)-y_{ij}(0)$ is set to zero and hence the exclusion restriction holds. A sensitivity analysis proceeds by finding the worst-case inference for a particular $\Gamma$ in \eqref{eqn:sensi}. One then iteratively increases the value of $\Gamma$ until the null hypothesis can no longer be rejected. \citet{Baiocchi:2010} describe how McNemar's test can be used to test $H_{prop}^{(0)}$ with binary outcomes. 
 
The proportional dose model implies that all individuals who would comply with their assigned encouragement have an identical treatment effect $\lambda$, which precludes essential heterogeneity among other forms of effect heterogeneity. Furthermore, with binary outcomes, the only plausible value for this effect would be zero: $y_{ij}(1)-y_{ij}(0)$ can only take the values -1, 0, and 1, with $\lambda =-1$ or 1 reflecting an extremely strong effect. See \citet[\S 5.1]{Rosenbaum:2002} for a related discussion on the restrictiveness of the constant effect model with binary outcomes outside of encouragement designs. We instead develop a sensitivity analysis for the weaker null hypothesis \begin{align*} H_{weak}^{(\lambda_0)}: \frac{\sum_{i=1}^{n}\sum_{j=1}^{2} \{y_{ij}(1) - y_{ij}(0)\}}{\sum_{i=1}^{n}\sum_{j=1}^{2} \{d_{ij}(1) - d_{ij}(0)\}} = \lambda_0\end{align*} while leaving individual effects unspecified. The proportional dose hypothesis $H_{prop}^{(\lambda_0)}$ is simply one element of $H_{weak}^{(\lambda_0)}$, and need not yield the supremum $p$-value over the composite null. 

\subsection{A valid approach with heterogeneous effects}
\label{sec:senshet}
 Observe that under the null hypothesis, $(2n)^{-1}\sum_{i=1}^n\sum_{j=1}^2[y_{ij}(1)-y_{ij}(0) - \lambda_0\{d_{ij}(1)-d_{ij}(0)\}] = 0$. Define $\zeta^{(\lambda_0)}_{i} = (Z_{i1}-Z_{i2})\{Y_{i1}-Y_{i2} - \lambda_0(D_{i1}-D_{i2})\}$ as the encouraged-minus-non encouraged differences in the terms $Y_{ij}-\lambda_0D_{ij}$, which may be thought of as the observed outcome adjusted for the treatment level received for each individual $ij$. Observe that at $\Gamma=1$, $E(\zeta_i^{(\lambda_0)}\mid \cF, \cZ) = (1/2)\sum_{j=1}^2[y_{ij}(1)-y_{ij}(0) - \lambda_0\{d_{ij}(1)-d_{ij}(0)\}]$, and that $n^{-1}\sum_{i=1}^n\zeta_i^{(\lambda_0)}$ hence forms an unbiased estimating equation for $\lambda_0$ in a randomized encouragement design.  Further observe that under $H_{prop}^{(\lambda_0)}$, the absolute value $|\zeta_i^{(\lambda_0)}|$ is fixed across randomizations, with only its sign varying. This observation underpins the use of permutation-based methods for sensitivity inference described in \citet{Rosenbaum:2002}; however, for other elements of $H_{weak}^{(\lambda_0)}$, the value of $|\zeta_i^{(\lambda_0)}|$ may itself be random over $\bz\in \Omega$.

Suppose one wants to conduct a sensitivity analysis for whether the effect ratio equals $\lambda_0$ at level of unmeasured confounding $\Gamma$ without assuming proportional doses. Define $L_{\Gamma i}$ as\begin{align*}L_{\Gamma i}= \zeta^{(\lambda_0)}_{i} - \left(\frac{\Gamma-1}{1+\Gamma}\right)|\zeta^{(\lambda_0)}_{i}|.
\end{align*} 
The term $\{(\Gamma-1)/(1+\Gamma)\}|\zeta^{(\lambda_0)}_{i}|$ is the worst-case expectation for $\zeta^{(\lambda_0)}_i$ at $\Gamma$ if the proportional dose model actually held. As proven in \citet{Fogarty:2019}, this provides an upper bound even if effects are heterogeneous.  Based upon the $n$ random variables $L_{\Gamma i}$, define $\text{se}(\bar{L}_\Gamma)$ to be the conventional standard error for a paired design, $\text{se}(\bar{L}_\Gamma)^2 = \{n(n-1)\}^{-1}\sum_{i=1}^n(L_{\Gamma i} - \bar{L}_\Gamma)^2$, and consider using as a test statistic with a greater-than alternative $T(\mathbf{Z}, \mathbf{Y}-\lambda_0\mathbf{D}) =  \bar{L}_\Gamma/\text{se}(\bar{L}_\Gamma)$. 

We now construct a reference distribution for the sensitivity analysis. When testing with a greater-than alternative, the sensitivity analysis bounds the upper-tail probability for the employed test statistic. Let $V_{\Gamma i}$ $(i=1,...,n)$ be conditionally independent given $\cF, \cZ$ and take values $\pm 1$ with $\P(V_{\Gamma i} = 1\mid \cF, \cZ) = \Gamma/(1+\Gamma)$, and define the random variable $B_{\Gamma i}= V_{\Gamma i}|\zeta^{(\lambda_0)}_{i}| - \{(\Gamma-1)/(1+\Gamma)\}|\zeta^{(\lambda_0)}_{i}|.$ The construction of $B_{\Gamma i}$ is motivated by the proportional dose model: under $H_{prop}^{(\lambda_0)}$, $B_{\Gamma i}$ stochastically dominates $L_{\Gamma i}$ if (\ref{eqn:sensi}) holds at $\Gamma$. Consider using as a reference distribution the distribution of $\bar{B}_\Gamma/\text{se}(\bar{B}_\Gamma)$ given $|\boldsymbol\zeta^{(\lambda_0)}|$; call its distribution function $G_\Gamma(\cdot)$. The decision whether or not to reject the null hypothesis in a sensitivity analysis at $\Gamma$ using the studentized statistic $\bar{L}_\Gamma/\text{se}(\bar{L}_\Gamma)$ is \begin{align}\label{eq:procedure}
\varphi^{(\lambda_0)}(\alpha,\Gamma) &= \1\{\bar{L}_\Gamma/\text{se}(\bar{L}_\Gamma) \geq G_\Gamma^{-1}(1-\alpha)\},
\end{align}  
where $G_\Gamma^{-1}(q) = \inf \{k : G_\Gamma(k) \geq q\}$ is the $q$th quantile of the distribution of $\bar{B}_\Gamma/\text{se}(\bar{B}_\Gamma)$ given $|\boldsymbol\zeta^{(\lambda_0)}|$. See Algorithm 2 of \citet{Fogarty:2019} for more on constructing $G_\Gamma(\cdot)$.\begin{proposition}Suppose that $H_{weak}^{(\lambda_0)}$ is true and that (\ref{eqn:sensi}) holds at $\Gamma$. Then, under mild regularity conditions and for $\alpha \leq 0.5$, ${\lim}_{n\rightarrow \infty}\; E\{\varphi^{(\lambda_0)}(\alpha,\Gamma)\mid \cF, \cZ\}\leq \alpha$.\label{prop:stu}\end{proposition}\vspace{-.2 in}The proof of Proposition \ref{prop:stu} mirrors that of Theorem 2 in \citet{Fogarty:2019} and is sketched in the web-based supplement. For testing the effect ratio without assuming proportional doses, regularity conditions are required to ensure that $n\times \text{se}(\bar{L}_\Gamma)^2$ converges in probability to a limiting value, and that a central limit theorem applies to $\sqrt{n}\{\bar{L}_\Gamma - E(\bar{L}_\Gamma\mid \cF, \cZ)\}$.
\begin{remark} At $\Gamma=1$, our method is asymptotically equivalent to that of \citet{Baiocchi:2010}. The test statistic proposed in that work for inference at $\Gamma=1$ is precisely $\bar{L}_1/\text{se}(\bar{L}_1)$, and the methods differ solely in the reference distribution employed. \citet{Baiocchi:2010} use the Normal distribution, while our approach uses a reference distribution generated by biased permutations of the observed data under the assumption of proportional doses in order to maintain exactness under that model for effects. Under suitable regularity conditions our reference distribution $G_\Gamma(k)$ converges in probability to $\Phi(k)$, the CDF of a standard Normal, pointwise at all points $k$. Their method can be viewed as a large-sample approximation of our method at $\Gamma=1$, and replacing our reference distribution $G_\Gamma(k)$ at $\Gamma > 1$ with the standard Normal $\Phi(k)$ would provide the natural extension of their large-sample approach for inference on the effect ratio to a sensitivity analysis.

\end{remark}

\subsection{Binary responses: Equivalence with McNemar's test and testing nonzero nulls}

\citet{Baiocchi:2010} suggest conducting a sensitivity analysis using McNemar's test while restricting attention to the narrower null hypothesis $H_{prop}^{(0)}$, amounting to a test of Fisher's sharp null of no effect. One may be concerned that this sensitivity analysis would be misleading if instead $H_{weak}^{(0)}$ held, such that the effect ratio equaled zero but Fisher's sharp null did not hold. As we now demonstrate, McNemar's test is actually equivalent to the test $\varphi^{(0)}(\alpha,\Gamma)$ presented in the previous section with $\lambda_0=0$, and hence also provides a sensitivity analysis that is both finite-sample exact for $H_{prop}^{(0)}$ and asymptotically correct for $H_{weak}^{(0)}$. 
\begin{proposition}\label{thm.equiv} 
Suppose that outcomes are binary and one employs a sensitivity analysis based on McNemar's test statistic, $T_{M}(\mathbf{Z}, \mathbf{Y}) = \sum_{i=1}^n\sum_{j=1}^2Z_{ij}Y_{ij},$ using its worst-case distribution under the assumption of $H^{(0)}_{prop}$ as described in \citet{Rosenbaum:1987} and \citet{Baiocchi:2010}. Denote the resulting candidate level$-\alpha$ sensitivity analysis by $\varphi_M(\alpha, \Gamma)$. Then, for any observed vector $\bZ \in \Omega$ and any vector of binary responses $\bY$, $\varphi_M(\alpha, \Gamma) = \varphi^{(0)}(\alpha,\Gamma).$ That is, the two sensitivity analyses are equivalent, furnishing identical $p$-values for all $\Gamma$.
\end{proposition}

The proof of Proposition \ref{thm.equiv} is presented in the web-based supplement, and requires showing that over elements of $\Omega$, $\bar{L}_\Gamma/\text{se}(\bar{L}_\Gamma)$ is a strictly increasing function of McNemar's statistic $T_M(\bZ, \bY) = \sum_{i=1}^nZ_{ij}Y_{ij}$, the number of encouraged individuals for whom an event occurred. Because $\varphi^{(0)}(\alpha,\Gamma)$ is valid under the weak null of no effect ratio, the sensitivity analysis for binary outcomes presented in  \citet{Baiocchi:2010}, motivated under Fisher's sharp null, is also an asymptotically correct sensitivity analysis for the effect ratio equaling zero in the presence of effect heterogeneity. Sensitivity analyses using McNemar's tests are ubiquitous in paired observational studies with binary outcomes, and Proposition \ref{thm.equiv} provides a useful fortification for those analyses should a critic be concerned about effect heterogeneity. 

The equivalence with McNemar's test pertains to the test that the effect ratio equals zero. The procedure $\varphi^{(\lambda_0)}(\alpha,\Gamma)$ remains asymptotically valid for testing $H_{weak}^{(\lambda_0)}$ for $\lambda_0\neq 0$ even with binary outcomes. Hence, the test  $\varphi^{(\lambda_0)}(\alpha,\Gamma)$ can be inverted to provide asymptotically valid sensitivity intervals for the effect ratio with binary outcomes.

\subsection{Application to the study of surgical outcomes for EGS patients}

The estimates in Table~\ref{tab.out} assumes that patients are assigned to surgeons in an as-if random fashion after accounting for observed covariates. Since our IV is not the result of a randomized encouragement, we cannot rule out the presence of hidden bias from IV--outcome confounders.  In Table~\ref{tab.sens}, we summarize the sensitivity analysis using the changepoint value -- the value of $\Gamma$ at which the estimate is no longer statistically significant at $\alpha=0.05$ \citep{Zhao:2019}.

\begin{table}[ht]
\centering
\caption{Sensitivity analysis for the effect of surgery on EGS patient outcomes. Cell entries are sensitivity values when testing at $\alpha=0.05$, the largest values of $\Gamma$ for which the null of zero effect ratio is rejected \citep{Zhao:2019}.}
\label{tab.sens}
\begin{tabular}{lcc}
  \toprule
 & Septic Patients & Non-septic Patients \\
\midrule
Complication  & 1.56 & 1.59  \\
     Length of Stay  & 1.64 & 2.36 \\ 
     Mortality  & 1.00 & 1.00 \\ 
\bottomrule
\end{tabular}
\end{table}

Despite the large point estimates for the complications outcome, we find that a somewhat modest amount of confounding could explain this result. For example, for non-septic patients an unobserved covariate would have to increase the odds of treatment by a surgeon with a high TTO by a factor of 1.59 within matched pairs in order to overturn our finding of a positive effect ratio at $\alpha = 0.05$. Table~\ref{tab.sens} illustrates that when effect modification is present, hypothesis tests for different subgroups may vary not only in their degree of statistical significance assuming $\Gamma=1$, but also in their sensitivity to bias should the test at $\Gamma=1$ reject. For related results, see also  \citet{lee2018discovering}. While the point estimates are larger for septic patients as displayed in Table \ref{tab.out}, the $\Gamma$ changepoint values are larger in the non-septic group for both complication and hospital length of stay. For septic patients an unobserved covariate would have to increase the odds of treatment by a surgeon with a high TTO by a factor of 1.64 within matched pairs in order to overturn our finding of a positive effect ratio, whereas a factor of 2.36 would be required for non-septic patients.

Sepsis is but one of many patient-level health characteristics which might influence whether or not surgery will be more effective than non-surgical alternatives. We next explore how accounting for additional heterogeneity that is predictable on the basis of effect modifiers may further decrease reported sensitivity to hidden bias. After doing so, we present theoretical results to help explain the pattern in Table \ref{tab.sens} of smaller treatment effects being considerably less sensitive to hidden bias.

\section{Improved standard errors by exploiting effect modification}
\label{sec:var}

\subsection{Conservativeness of finite population sensitivity analysis under effect heterogeneity}

Proposition \ref{prop:stu} tells us that rejecting  through the test $\varphi^{(\lambda_0)}(\alpha,\Gamma)$ provides an asymptotically valid sensitivity analysis, in the sense that if (\ref{eqn:sensi}) holds at $\Gamma$ then the test will commit a Type I error with probability at most $\alpha$ in the limit. In practice, it could be that $\varphi^{(\lambda_0)}(\alpha,\Gamma) \ll \alpha$ under $H_{weak}^{(\lambda_0)}$, resulting in a conservative test. Imagine that we had access to $E(\bar{L}_\Gamma\mid \cF, \cZ)$ and $\sd(\bar{L}_\Gamma\mid \cF, \cZ)$, the \textit{true} expectation and standard deviation of the conditional distribution $\bar{L}_\Gamma \mid \cF, \cZ$. In this case, rejecting the null hypothesis when the test statistic $\{\bar{L}_\Gamma - E(\bar{L}_\Gamma\mid \cF, \cZ)\}/\sd(\bar{L}_\Gamma\mid \cF, \cZ)$ exceeds $\Phi^{-1}(1-\alpha)$ would, under mild conditions ensuring that a central limit theorem holds, result in a test statistic with asymptotic level exactly equal to $\alpha$. Unfortunately, this test statistic cannot be deployed in practice. For $\Gamma > 1$ the expectation $E(\bar{L}_\Gamma\mid \cF, \cZ)$ is unknown to its dependence on both the vector of unmeasured confounders $\mathbf{u}$ and the missing potential outcomes, which are not imputed under $H_{weak}^{(\lambda_0)}$. Further, for any value of $\Gamma$, $\sd(\bar{L}_\Gamma\mid \cF, \cZ)$ is also unknown when effects are heterogeneous. By using the test statistic $\bar{L}_\Gamma/\text{se}(\bar{L}_\Gamma)$, our procedure overcomes these difficulties by means of the bounds $E(\bar{L}_\Gamma\mid \cF, \cZ) \leq 0$ and $\var(\bar{L}_\Gamma\mid \cF, \cZ) \leq E\{\se(\bar{L}_\Gamma)^2\mid \cF, \cZ\}$ \citep[Lemmas 2-3]{Fogarty:2019}. While we cannot generally hope to improve upon the bound $E(\bar{L}_\Gamma\mid \cF, \cZ)\leq 0$ for $\Gamma>1$ without risking an anti-conservative procedure, we illustrate that improvements on $\text{se}(\bar{L}_\Gamma)$ \textit{can} be attained while preserving the asymptotic level.

\subsection{A general construction of valid standard errors for sensitivity analysis}
Extending Lemma 3 of \citet{Fogarty:2019} to encouragement designs, we have \begin{align*}
E\{\text{se}(\bar{L}_\Gamma)^2\mid \cF, \cZ\} &= \var(\bar{L}_\Gamma \mid \cF, \cZ) + \frac{1}{n(n-1)}\sum_{i=1}^n(\mu_{\Gamma i} - \bar{\mu}_\Gamma)^2,
\end{align*}
where $\mu_{\Gamma i} = E(L_{\Gamma i}\mid \cF, \cZ)$ is the true, but unknowable, expectation of $L_{\Gamma i}$ and $\bar{\mu}_\Gamma = n^{-1}\sum_{i=1}^n\mu_{\Gamma i}$. The bias in $\text{se}(\bar{L}_\Gamma)$ thus depends upon the degree of heterogeneity in the expectations $\mu_{\Gamma i}$. We can interpret the magnitude of the bias as the mean squared error estimate from a regression of $\mu_{\Gamma i}$ on an intercept column $\mathbf{1}_{n}$, a vector containing $n$ ones. Imagine now that we were able to account for some of the variation in $\mu_{\Gamma i}$ through covariance adjustment. This would reduce the mean squared error, and hence the degree of bias.

Let $\mathbf{Q}$ be an $n \times p$ matrix with $p < n$ that is constant across all $\bz \in \Omega$, and let $\mathbf{H} = \bQ (\bQ^T \bQ)^{-1} \bQ^T$ be the orthogonal projection of $\mathbb{R}^n$ onto the column space of $\bQ$. Let $h_{ik}$ be the $\{i, k\}$ element of $\bH$ and let $\tilde{L}_{\Gamma i} = L_{ \Gamma i}/\sqrt{1-h_{ii}}$. Define $\text{se}(\bar{L}_\Gamma; \bQ)^2$ as\begin{align}
	\text{se}(\bar{L}_\Gamma; \bQ)^2 = \frac{1}{n^2} \mathbf{\tilde{L}_{\Gamma}}^T (\mathbf{I} -\mathbf{H}) \mathbf{\tilde{L}_{\Gamma}},
	\label{eqn:new_var}
\end{align}   
where $\bI$ is the $n \times n$ identity matrix and $\mathbf{\tilde{L}_\Gamma} = (\tilde{L}_{\Gamma 1}, \ldots, \tilde{L}_{\Gamma n})^T$. 

\begin{proposition}
For any value of $\Gamma$, $E\{\text{se}(\bar{L}_\Gamma; \bQ)^2 \mid \mathcal{F}, \mathcal{Z}\} - \var(\bar{L}_\Gamma \mid \mathcal{F}, \mathcal{Z}) = n^{-2} \bm{\tilde{{\mu}}_\Gamma}^T (\bI - \bH) \bm{\tilde{{\mu}}_\Gamma} \geq 0$,
	where $\bm{\tilde{{\mu}}_\Gamma} = (\tilde{\mu}_{\Gamma 1}, \ldots, \tilde{\mu}_{\Gamma n})^T$, and $\tilde{\mu}_{\Gamma i} = \mu_{\Gamma i}/\sqrt{1-h_{ii}}$.
	\label{prop:var1}
\end{proposition}The proof of Proposition \ref{prop:var1} follows that of Proposition 1 in \citet{Fogarty:2018}. Regardless of the true value of $\Gamma$ for which (\ref{eqn:sensi}) holds and regardless of the form for the matrix $\bQ$, $\text{se}(\bar{L}_\Gamma, \bQ)^2$ will be conservative in expectation for $\var(\bar{L}_\Gamma \mid \mathcal{F}, \mathcal{Z})$ so long as $\bQ$ does not vary across $\Omega$.

\subsection{A regression-based standard error} 
Setting $\bQ = \mathbf{1}_n$ recovers the usual standard error estimator. Guided by the form of the bias, we see that $\bQ$ should be chosen to predict the individual-level expectations $\tilde{\mu}_{\Gamma i} \approx \mu_{\Gamma i}$.  One choice would be to let $\bQ = (\mathbf{1}, \bar{\mathbf{X}})$ be a $n \times (k+1)$ matrix where the $q$th column of $\bar{\mathbf{X}}$ is $(\bar{x}_{1q}, \ldots, \bar{x}_{nq})^T$, and where $\bar{x}_{iq} = (x_{i1q} + x_{i2q})/2$, the average of the $q$th covariate's values in the $i$th pair. This reflects a hope that the heterogeneous expectations $\mu_{\Gamma i}$ may be linear in the within-pair covariate averages. Under mild regularity conditions, it can be shown that the modified variance estimator $\text{se}(\bar{L}_\Gamma; \bQ_{reg})$ with $\bQ_{reg}= (\mathbf{1}_n, \bar{\mathbf{X}})$ is asymptotically never worse than $\se(\bar{L}_\Gamma)$, \textit{regardless} of whether or not $\bm{\mu_\Gamma}$ is truly linear in the within-pair covariate averages $\bar{\mathbf{X}}$. 
\begin{proposition}
	Under regularity conditions and defining 0/0 = 1,\begin{align*}
	\frac{\text{se}(\bar{L}_\Gamma; \bQ_{reg})^2 - \var(\bar{L}_{\Gamma} \mid \mathcal{F}, \mathcal{Z})}{\se(\bar{L}_\Gamma; \mathbf{1}_n)^2 - \var(\bar{L}_{\Gamma} \mid \mathcal{F}, \mathcal{Z})}&\xrightarrow{p} 1-R^2,
	\end{align*}
	where $R^2$ is the coefficient of determination from a regression of $\bm{\mu_\Gamma}$ on $(\bm{1_n}, \bar{\mathbf X})$.
	\label{prop:var2}
\end{proposition}
Sufficient conditions for Proposition \ref{prop:var2} are presented along with its proof in the supplementary materials. The conclusion of Proposition \ref{prop:var2} would hold replacing $\bm{\tilde{L}_\Gamma}$ with $\bm{L_\Gamma}$, and with $1/n^2$ replaced by $1/\{n(n-k-1)\}$ in (\ref{eqn:new_var}). In that case, the resulting standard error estimator would simply be the $RMSE$ from a regression of $\bm{L_\Gamma}$ on $(\bm{1_n}, {\bar{\mathbf X}})$ divided by $\sqrt{n}$.

\subsection{A nonparametric approach: pairing the pairs through nonbipartite matching}\label{sec:pop}
While the regression-based standard error does not require a properly specified linear model for its validity, the gains are dependent on the extent to which $\bm{\mu_\Gamma}$ is linear in $\bQ$. While averages of polynomial terms can be included within $\bar{\bX}$, alternative nonparametric approaches may be preferred. We now illustrate that a variant of the approach considered in \citet{Abadie:2008} can also be employed in a finite population sensitivity analysis. 

Suppose $n$ is even and consider ``pairing the pairs" through the use of nonbipartite matching based on Mahalanobis distances between the average values of the within-pair covariates $\bar{\mathbf{x}}_i$. This results in $n/2$ ``pairs of pairs,'' where pairs have been matched to other pairs with similar average values for $\bar{\bx}_i$. For each pair, let $\mathcal{J}(i)$  be the index of the pair that was matched with the $i$th pair such that $\mathcal{J}(\mathcal{J}(i)) = i$. Consider the variance estimator \begin{align}\label{eq:pop}
\text{se}(\bar{L}_\Gamma; \bQ_{PoP})^2 &= \frac{1}{2n^2}\sum_{i=1}^n(L_{\Gamma i} - L_{\Gamma \mathcal{J}(i)})^2,
\end{align}which simply squares the differences in the terms $L_{\Gamma i}$ within each pair of pairs.\begin{proposition}\begin{align*}E\{\text{se}(\bar{L}_\Gamma; \bQ_{PoP})^2\mid \cF, \cZ\}- \var(\bar{L}_\Gamma\mid \cF, \cZ) &= \frac{1}{2n^2}\sum_{i=1}^n(\mu_{\Gamma i} - \mu_{\Gamma \mathcal{J}(i)})^2. \end{align*}\end{proposition}
\vspace{-.1 in}While a proof unique to this estimator follows by expanding the square, as we show in the web-based supplement the result can also be obtained through Proposition \ref{prop:var1} by letting $\bQ_{PoP}$ be the $n \times (n/2)$ matrix with the $q$th column containing membership indicators for the $q$th of $n/2$ pairs of pairs. See the supplementary material for how to use an odd number of pairs and for a comparison of the regression and pairs of pairs approaches.

\section{Improved sensitivity analysis }
\subsection{An improved sensitivity analysis for sample average effects}
\label{sec:new}

The standard errors in \S \ref{sec:var} can be used to improve the power of the sensitivity analysis while maintaining the asymptotic level of the test. Consider conducting a sensitivity analysis using $\bar{L}_\Gamma/\text{se}(\bar{L}_\Gamma; \bQ)$ as a test statistic for some matrix $\bQ$.  Let $G_{\Gamma}(\cdot; \bQ)$ be the distribution of $\bar{B}_\Gamma/\text{se}(\bar{B}_\Gamma; \bQ)$ given $|\bm{\zeta}^{(\lambda_0)}|$, where $\bar{B}_\Gamma$ is defined as in \S \ref{sec:senshet}. The modified test is
\begin{align*}
\varphi_\bQ^{(\lambda_0)}(\alpha,\Gamma ) &= \1\{\bar{L}_\Gamma/\text{se}(\bar{L}_\Gamma; \bQ) \geq G_\Gamma^{-1}(1-\alpha; \bQ)\}.
\end{align*}
\begin{proposition}\label{prop:improveinf}
Suppose (\ref{eqn:sensi}) holds at $\Gamma$ and that $H_{weak}^{(\lambda_0)}$ holds. Under regularity conditions\;\;${\lim}_{n\rightarrow\infty}\;E\{\varphi_\bQ^{(\lambda_0)}(\alpha,\Gamma )\} \leq \alpha$, such that the procedure yields a valid level$-\alpha$ sensitivity analysis for the effect ratio with heterogeneous effects. Furthermore, for $\bQ = \bQ_{reg}$ or $\bQ = \bQ_{PoP}$ and regardless of whether or not the null holds, then under regularity conditions 
$\lim_{n\rightarrow\infty}\;E\{\varphi_\bQ^{(\lambda_0)}(\alpha,\Gamma )\} \geq {\lim}_{n\rightarrow\infty}\; E\{\varphi^{(\lambda_0)}(\alpha,\Gamma )\},$
such that the new procedure is both less conservative under the null and more powerful under the alternative than the procedure employing the conventional standard error estimator. 
\end{proposition}
\vspace{-.1 in}
The magnitude of the improvement will depend upon the extent to which $\text{se}(\bar{L}_\Gamma;\bQ)$ improves upon the conventional standard error estimator, $\text{se}(\bar{L}_\Gamma)$. In the web-based supplement, we present a detailed simulation study illustrating the potential benefits conferred by these improved standard error estimators when heterogeneous effects are present. The simulations include scenarios both with no hidden bias and where hidden bias is present. The web-based supplement also includes a sketch of the proof of Proposition \ref{prop:improveinf}, a discussion of sufficient conditions, and an algorithm illustrating the procedure $\varphi_\bQ^{(\lambda_0)}(\alpha,\Gamma )$.

\subsection{Application to the EGS study}

Table~\ref{tab.sens2} contains a comparison of 95\% confidence intervals calculated at $\Gamma=1$ and sensitivity values within the septic and non-septic subgroups for the hospital length of stay outcome using three standard errors: the conventional standard error for a paired design, the regression-based standard error, and one formed by pairing together pairs on the basis of observed covariates. In our application, the different standard error estimators did not materially alter the confidence interval lengths. We found that the $\Gamma$ changepoints were also similar across standard errors. Results for the complication and mortality outcomes indicated a similar trend and are omitted here. 

\begin{table}[ht]
\centering
\caption{Confidence intervals and sensitivity values for hospital length of stay with different standard errors. }
\label{tab.sens2}
\begin{tabular}{llclc}
  \toprule
& \multicolumn{2}{c}{Septic Patients}&\multicolumn{2}{c}{Non-septic Patients}  \\
\cmidrule(lr){2-3}\cmidrule(lr){4-5}
Standard Error & 95\% Conf. Int.  & Sens. Value & 95\% Conf. Int.  & Sens. Value\\
\midrule
Conventional & [ 3.87 , 8.71 ] & 1.64& [ 3.36 , 4.81 ] & 2.36\\ 
Regression-based  & [ 4.03 , 8.69 ]&  1.62& [ 3.36 , 4.81 ]& 2.38 \\ 
Pair of Pairs  & [ 4.01 , 8.66 ] & 1.64& [ 3.37 , 4.82 ]& 2.37\\ 
\bottomrule
\end{tabular}
\end{table}This lack of improvement suggests that there is little evidence for substantive effect modification within the sepsis and non-sepsis subgroups on the basis of other patient characteristics. In the web-based supplement, we formally derive an exact omnibus test for effect modification which leverages our improved standard errors. Using this test, for all outcome variables and within both the septic and non-septic subgroup we fail to reject the null hypothesis that the proportional dose model in \S \ref{sec:propdose} holds, finding little evidence that the IV estimates vary as a function
of over 90 additional observed covariates.
\section{Heterogeneity and the power of a sensitivity analysis in IV studies}
\subsection{Larger but brittler effect estimates}
While there was no evidence for effect modification within the septic and non-septic subgroups on the basis of observed covariates, the effect ratio estimates within the septic group were larger than those in the non-septic subgroup for all three outcomes as presented in Table \ref{tab.out}. For the complication outcome, we were able to reject the null of no effect ratio until $\Gamma=1.56$ and $\Gamma=1.59$ for the septic and non-septic groups respectively. For the length of stay variable, the discrepancy is even more substantial: for two individuals in the same matched set, a hidden variable would need to produce discrepancies in the odds that the individuals were assigned to a high TTO surgeon by a factor of 1.64 in the septic group, while in the septic-group the odds would need to differ by a factor of 2.36. Despite the septic subgroup having the larger estimated effects, the estimates for the \textit{non-septic} group proved to be more robust to hidden bias. We now demonstrate how this phenomenon is reflective of the importance of reduced within-pair homogeneity for reducing sensitivity to hidden bias.
\subsection{A favorable reality unknown to the practitioner}
\label{sec:favorable}
Imagine that $Z_{ij}$ were actually a valid instrument and that the true value for the effect ratio exceeded its hypothesized value. Given that our instrument is not randomly assigned, a critic could always counter that a large effect estimate is merely due to bias from lurking variables invalidating the proposed instrument. Even in this favorable situation of a valid instrument and a positive treatment effect, we would hope that our inferences would prove robust to moderate degrees of hidden bias to protect ourselves against such criticism. Calculations within this section proceed under this favorable setting: there is no hidden bias, there is truly an effect, but the practitioner, blind to this reality, hopes that her inferences perform well under the stress of a sensitivity analysis. We thus assess our method's ability to discriminate between (1) no treatment effect and hidden bias; and (2) a treatment effect without hidden bias.

Until this point inference has been performed conditional upon $\cF$, and a generative model for $\cF$ has been neither required nor assumed. For the calculations in this section, it is convenient to assume a superpopulation model. Because $\cF$ will be viewed as random, the procedures exploiting effect modification presented in \S \ref{sec:new} no longer provide valid tests for the effect ratio \citep[\S 5]{Fogarty:2018}. We focus on the sensitivity analysis given in (\ref{eq:procedure}) using the conventional standard error estimate, $\varphi^{(\lambda_0)}$, which remains valid under a superpopulation model. Following \citet{Small:2008}, we imagine that $\zeta_i^{(\lambda_0)}$ is generated as \begin{align}\label{eq:gen2}
\zeta_i^{(\lambda_0)} &= \epsilon_i + S_i(\lambda - \lambda_0),
\end{align} 
where $\epsilon_i = (Z_{i1}-Z_{i2})\{(Y_{i1}-Y_{i2}) - \lambda(D_{i1} - D_{i2})\}$ are the adjusted encouraged-minus-non encouraged differences in responses, and $S_i = (Z_{i1}-Z_{i2})(D_{i1}-D_{i2})$ are the encouraged-minus-non encouraged differences in the treatment received, reflecting the strength of the instrument. Note that this generative model does \textit{not} imply the proportional dose model, such that the individual-level effects are allowed to be heterogeneous.

We assume that $\epsilon_i$ are $iid$ from a symmetric distribution with mean zero and finite variance $\sigma^2$. The treatments received are assumed binary. We assume that there are no defiers, that the exclusion restriction holds, and that individuals $ij$ are assigned status as compliers, never-takers and always-takers independently with probability $p_C$, $p_N$, and $p_A$ respectively. This results in $\P(S_i=1) = p_C + p_Ap_N$, $\P(S_i=-1) = p_Ap_N$ and $\P(S_i=0) = 1-p_C-2p_Ap_N$. The true treatment effect among compliers is $\lambda$, while $\lambda_0$ is its value under the null.

\subsection{Making sense of our sensitivity analysis}\label{sec:powersim}

Through simulation, we now compare the role that the effect size relative to individual-level variability plays in an observational study. We use the length of stay outcome to motivate parameter values, assuming that the estimated effect ratios are actually the true values of $\lambda$ and that the standard deviations of $\zeta_i^{(\lambda_S)}$ and $\zeta_i^{(\lambda_{NS})}$ in our sample reflect the true standard deviations for these distribution. We assume there are no defiers, and set the probability of compliance to $p_C = 0.58$ (the estimated compliance rate in our data set), and set $p_A=p_N = 0.21$. We simulate from the generative model (\ref{eq:gen2}) using different sample sizes for the septic and non-septic groups. The ratio of the variances of $\epsilon_i$ in the septic and non-septic subgroups is $\sigma^2_S/\sigma^2_{NS}\approx 8$, and we set $n_S/n_{NS}$, the ratio of the sample sizes for the septic and non-septic groups, to also equal 8. We do this so that the variances of the sample averages of the $\epsilon_i$ terms in both groups, $\sigma^2_S/n_S$ and $\sigma^2_{NS}/n_{NS}$, are approximately equal. The equality of these variances for the sample mean may appear to put inference within the non-septic subgroup at a disadvantage: as $\lambda_S > \lambda_{NS}$, inference at $\Gamma=1$ for the null of no effect ratio would be more powerful in the septic group. Does this disadvantage carry over to a sensitivity analysis at $\Gamma > 1$?

In each simulation, $n_S$ and $n_{NS}$ observations are drawn from the distribution (\ref{eq:gen2}), with parameters $(\lambda_{S}, \sigma^2_S)$ and $(\lambda_{NS}, \sigma^2_{NS})$ and with $\epsilon_i$ normally distributed. After drawing the samples, we conduct a sensitivity analysis at $\Gamma$, and record whether or not we correctly reject the null of zero effect ratio for each subgroup. We proceed with $n_S=200, 1000, 2000$, $10,000$ and $n_{NS} = n_S/8$ over a range of $\Gamma$ values. We run 10,000 simulations for each setting. Through this simulation, we highlight the different roles that sample size plays for inference assuming no hidden bias ($\Gamma=1$) and in a sensitivity analysis. \begin{figure}
\begin{center}
\includegraphics[scale=.45]{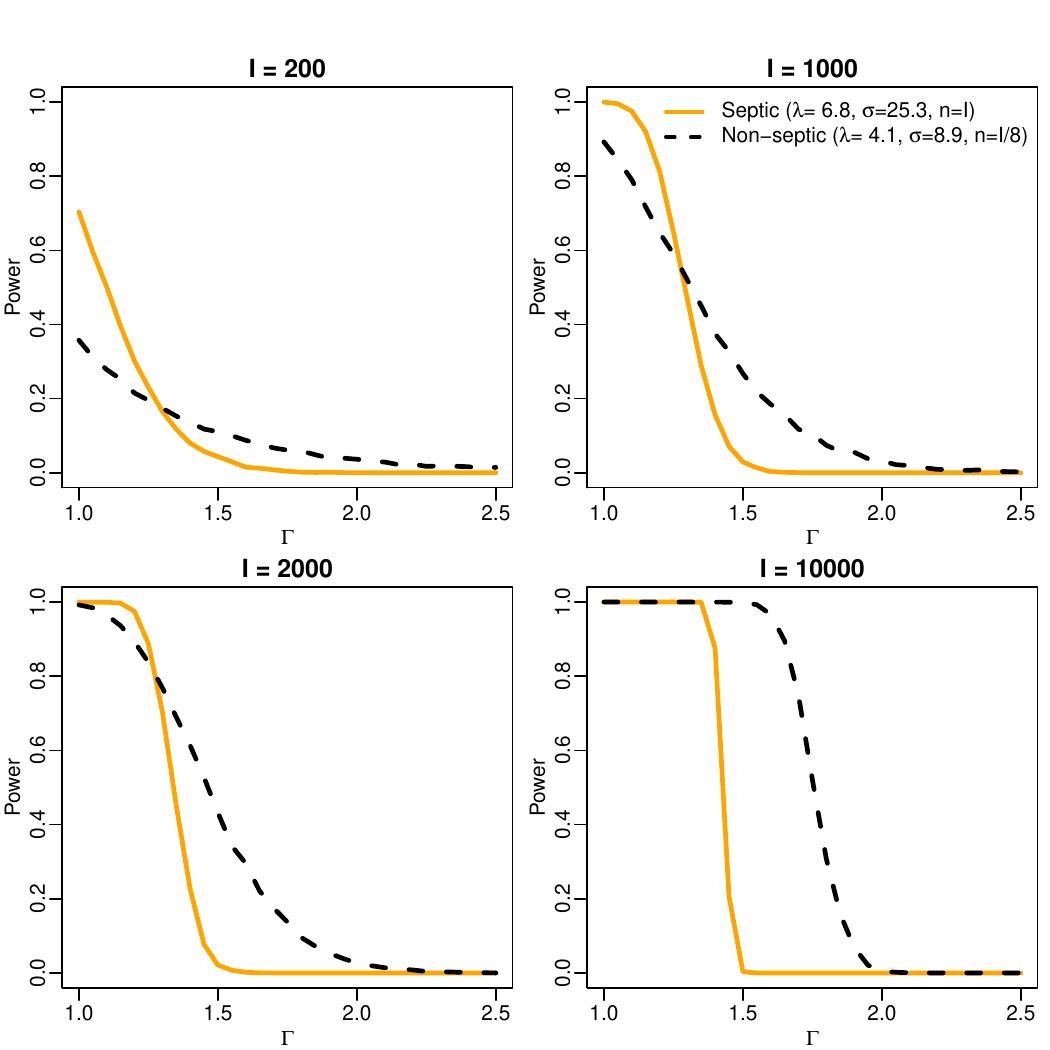}
\caption{The four plots show the power of the test against the null $\lambda = 0$ in the absence of unmeasured confounding as a function of $\Gamma$ for $I = 200, 1000, 2000$, and $10000$. The septic subgroup has $n_S = I$ pairs with a larger effect ratio and larger degree of heterogeneity, while the non-septic has $n_{NS} = I/8$ pairs, a smaller effect ratio and a smaller degree of heterogeneity. The probability of compliance is 0.58 in all simulations.  }\label{fig:power}
\end{center}
\end{figure}
Figure \ref{fig:power} presents the results of the simulation study. Each plot shows the performance of the sensitivity analyses in the septic and non-septic subgroups, varying $n_S$ while maintaining $n_S/n_{NS} = 8$. For instance, the upper-left panel shows the performance in a sensitivity analysis as a function of $\Gamma$ with $n_S = 200$, and $n_{NS} = 25$. Assuming no hidden bias ($\Gamma=1$), the test in the septic subgroup has higher power as $\lambda_S > \lambda_{NS}$ and $\sigma^2_{NS}/n_{NS} \approx \sigma^2_{S}/n_S$. As $\Gamma$ increases, we see that despite the smaller effect size the test in the non-septic subgroup begins to outperform that in the septic group, owing to the fact that $\sigma^2_{NS} < \sigma^2_S$. As $n_S$ increases with $n_S/n_{NS}$ fixed, we see that the larger sample size in the septic-group provides a smaller and smaller benefit. In sufficiently large samples, at $\Gamma=1$ the tests in both subgroups reject the null of no effect ratio with near certainty. As $\Gamma$ increases in this larger sample regime, the power functions of the tests in both subgroups converge to step functions, with changepoints at $\Gamma=1.91$ for the non-septic group and $\Gamma=1.47$ for the septic group. These values are the \textit{design sensitivities} \citep{Rosenbaum:2004b} for these two generative models: asymptotically a test of no effect will reject with probability zero when $\Gamma$ is above the design sensitivity and one below it; see the web-based supporting material for a formula for the design sensitivities along with a numerical evaluation of the impact of compliance on design sensitivity in IV designs. Overall, the comparative performance between the septic and non-septic subgroups highlights the importance of reduced heterogeneity of responses within pairs for increasing robustness \citep{Rosenbaum:2005b}.

\section{Summary}

Assuming the validity of surgeon's tendency to operate as an instrument, our effect estimates indicate that surgery has an adverse, statistically significant, effect on both hospital length of stay and presence of a complication within both the septic and non-septic subgroups, with a larger estimated effect within the septic group. While the effect estimates for 30-day mortality indicate that surgery may also have an adverse effect, these results were not statistically significant at $\alpha=0.05$ even assuming no hidden bias ($\Gamma=1$). The complication outcome was robust to a moderate degree of hidden bias: for both septic and non-septic patients, two individuals with the same observed covariates would have to differ in their odds of being assigned to a high versus low TTO surgeon by a factor of roughly 1.6 to overturn the findings of the study. For septic patients a similar odds ratio discrepancy would be required to overturn the finding of an adverse effect on hospital length of stay, while for non-septic patients the odds would have to differ by nearly 2.4. This helps to frame the debate about what arguments against the validity of TTO as an instrument would actually matter: if a pattern of hidden bias could not influence the odds of assignment of surgeons to patients to this extent, it could not explain away the finding of an effect.

These conclusions required innovations in sensitivity analyses for IV designs. Given surgery's suspected variation in efficacy due to baseline patient characteristics we developed a sensitivity analysis that does not require effect homogeneity. Instead, we show how a sensitivity analysis for the effect ratio may proceed without specifying a pattern of effect heterogeneity. We then developed variance estimators that can result in improved sensitivity analyses when effects vary with observed covariates. Motivated by the finding that treatment effect estimates in the non-septic group were more robust to hidden bias despite being of a smaller estimated magnitude, we showcased the role of strong instruments and reduced within-pair effect variation for design sensitivity and the power of a sensitivity analysis.

\newpage
\appendix
\section{Review}
\setcounter{table}{4}

\begin{table}[h]
	\centering
	\begin{tabular}{ll}
		\toprule
		Notation & Description \\
		\midrule
		$Z_{ij}$ & Encouragement to treatment (1 yes, 0 no)\\
		$d_{ij}(1)$ \& $d_{ij}(0)$ & Potential exposures for individual $ij$, $i=1, \ldots n$; $j=1, 2$\\
		$y_{ij}(1)$ \& $y_{ij}(0)$ & Potential outcomes\\
		$Y_{ij}$ \& $D_{ij}$ & Observed outcome and exposure \\
		$\lambda$ & Effect ratio: \\
		& $\frac{\sum_{i=1}^{n} \sum_{j=1}^{2} \{ y_{ij}(1) - y_{ij}(0)\}}{\sum_{i=1}^{n} \sum_{j=1}^{2} \{ d_{ij}(1) - d_{ij}(0)\}} $\\
		$H_{prop}^{(\lambda_0)}$ & Null hypothesis under the proportional dose model \\
		& $y_{ij}(1) - y_{ij}(0) = \lambda_0 \{d_{ij}(1) - d_{ij}(0)\}\;\;\forall \;\;ij$ \\
		$H_{weak}^{(\lambda_0)}$ & Neyman's weak null hypothesis on the effect ratio\\ &$\lambda = \lambda_0$ \\
		\midrule
		\midrule
		\multicolumn{2}{l}{Definition}\\
		\midrule
		\multicolumn{2}{l}{$\zeta_{i}^{(\lambda_0)} = (Z_{i1} - Z_{i2})\{Y_{i1} - Y_{i2} - \lambda_0 (D_{i1} - D_{i2})\}$} \\[0.2cm]
		\multicolumn{2}{l}{$L_{\Gamma i} = \zeta_i^{(\lambda_0)} - \left(\frac{\Gamma -1}{\Gamma+1} \right) |\zeta_i^{(\lambda_0)}|$}\\[0.2cm]
		\multicolumn{2}{l}{$\mu_{\Gamma i} = E(L_{\Gamma i}\mid \cF, \cZ)$}\\[0.2cm]
		\multicolumn{2}{l}{$\se(\bar{L}_{\Gamma})^2 = \frac{1}{n(n-1)} \sum_{i=1}^{n} (L_{\Gamma i} - \bar{L}_{\Gamma})^2$}\\[0.2cm]
		
		\multicolumn{2}{l}{$\theta_\Gamma = \Gamma/(1+\Gamma)$}\\[0.2cm]
		\multicolumn{2}{l}{$B_{\Gamma i} = V_{\Gamma i} |\zeta_i^{(\lambda_0)}| - \left(\frac{\Gamma -1}{\Gamma+1} \right) |\zeta_i^{(\lambda_0)}|$}\\[0.2cm]
		\multicolumn{2}{l}{where $V_{\Gamma i} = \pm 1$; $\pr(V_{\Gamma i}= 1)= \theta_{\Gamma}$ and $\pr(V_{\Gamma i} = -1) = 1-\theta_{\Gamma}$} \\ [0.2cm]
		\multicolumn{2}{l}{$\varphi^{(\lambda_0)}(\alpha, \Gamma) = 1\{\bar{L}_{\Gamma}/\se(\bar{L}_{\Gamma}) \geq G_{\Gamma}^{-1}(1-\alpha)\}$}\\[0.2cm]
		\multicolumn{2}{l}{$\varphi_{\bQ}^{(\lambda_0)}(\alpha, \Gamma) = 1\{\bar{L}_{\Gamma}/\se(\bar{L}_{\Gamma}; \bQ) \geq G_{\Gamma}^{-1}(1-\alpha; \bQ)\}$}\\
		\multicolumn{2}{l}{$G_{\Gamma}(\cdot\; \bQ)$: Studentized reference distribution using matrix $\bQ$}\\
		\bottomrule
	\end{tabular}
\end{table}

\newpage
\subsection{Preliminaries}\label{sec:A1}

For each $\lambda_0$, we define a new variable $y_{ij}^*(z) = y_{ij}(z) - \lambda_0 d_{ij}(z) (z=0,1)$ which can be thought as the potential outcome adjusted for the dosage/exposure level received for each individual $ij$. The adjusted potential outcomes $y_{ij}^*(1), y_{ij}^*(0)$ are treated as if they are the potential outcomes usually defined in the literature. Define $\tau_{ij} = y_{ij}^*(1) - y_{ij}^*(0)$ and $\bar{\tau}_{i} = (\tau_{i1} + \tau_{i2})/2$. Then, $\tau_{ij}$ can be considered as the effect of the treatment on the adjusted outcome for individual $ij$, and $\bar{\tau_i}$ can be considered as the average treatment effect within pair $i$. Recall that $\zeta_i^{(\lambda_0)} = (Z_{i1} - Z_{i2})\{Y_{i1} - Y_{i2} - \lambda_0 (D_{i1} -D_{i2})\}$, which can be represented by $(Z_{i1} - Z_{i2})(Y_{i1}^* - Y_{i2}^*)$ where  $Y_{ij}^* = y_{ij}^*(1) Z_{ij} + y_{ij}^*(0) (1-Z_{ij})$ is the observed ``adjusted'' outcome. This quantity can be understood as the treated-minus-control paired ``adjusted'' outcome difference. Thus, $\zeta_i^{(\lambda_0)}$ would be an unbiased estimator for $\bar{\tau}_i$ in a randomized encouragement design, but may exhibit bias in the presence of unmeasured confounders. Further define
$$
\eta_{i} = \frac{y_{i1}^*(1) + y_{i1}^*(0)}{2} - \frac{y_{i2}^*(1) + y_{i2}^*(0)}{2}, 
$$
which is the difference in the averages of the potential outcomes for individuals within a pair. Also, recall that our sensitivity analysis model assumes, for each $\Gamma$, 
$$
1 - \theta_\Gamma \leq \pi_{i} = \pr(Z_{i1}=1 \mid \mathcal{F}, \mathcal{Z}) \leq \theta_\Gamma\;\;\; i=1, \ldots, n
$$
where $\theta_\Gamma = \Gamma/(1+\Gamma)$.

Asymptotics in the forthcoming discussions have in mind a single population $\cF$ of increasing size; randomness stems only from the assignments $\bZ$. The statement that $H_{weak}^{(\lambda_0)}$ holds in the limit should, in reality, reflect the changing value of the effect ratio $\lambda_{0n}$ as the finite population $\mathcal{F}_n$ grows. We will suppress both this dependence and conditioning upon $\cF$ and $\cZ$ for readability.

\subsection{Regularity conditions}

\begin{condition}
	There exist constants $C>0$, $\mu_m$ and $\mu_a$ such that as $n \to \infty$, 
	\begin{align}
		n^{-1} \sum_{i=1}^{n} |\eta_i| > C, \quad n^{-1}\sum_{i=1}^{n} \eta_i^2 > C, &\\
		n^{-2} \sum_{i=1}^{n}\eta_i^2 \to 0, \quad n^{-2} \sum_{i=1}^{n}\eta_i^4 \to 0, \quad n^{-2} \sum_{i=1}^{n}\bar{\tau}_i^4 \to 0,& \\
		n^{-1} \sum_{i=1}^{n} (2\pi_i -1) \eta_i \to \mu_m, \quad n^{-1} \sum_{i=1}^{n} \pi_i |\bar{\tau}_i + \eta_i| + (1-\pi_i)|\bar{\tau}_i - \eta_i| \to \mu_a &
	\end{align}
\label{eqn:condi1}
\end{condition}

\begin{condition}
	There exists a constant $\nu^2 > 0$ such that 
	\begin{align}
		n^{-1} \sum_{i=1}^{n} \pi_i (\bar{\tau}_i + \eta_i)^2 + (1-\pi_i)(\bar{\tau}_i - \eta_i)^2 \to \nu^2.
	\end{align}
\label{eqn:condi2}
\end{condition}

\section{Proofs}

\subsection{Proposition 1}

Given Conditions \ref{eqn:condi1} and \ref{eqn:condi2}, Proposition 1 follows immediately from Theorems 1 and 2 of \citet{Fogarty:2019}. We present a sketch of the proof here, the sketch slightly diverging from the proof in \citet{Fogarty:2019} as here we merely establish asymptotic Type I error control under the weak null. 

The first component of the proof ignores the reference distribution $G_\Gamma (\cdot)$, instead focusing on the large-sample reference distribution $\Phi(\cdot)$, the CDF of the standard Normal. Under Conditions \ref{eqn:condi1} - \ref{eqn:condi2}:

\begin{enumerate}
\item[(i)] Under $H_{weak}^{(\lambda_0)}$, the random variable $\sqrt{n}\bar{L}_\Gamma$ converges in distribution to a Normal random variable, whose expectation and standard deviation are unknown due to their dependence on both the unknown potential outcomes and the unmeasured confounders.

\item[(ii)] Under $H_{weak}^{(\lambda_0)}$ and if the sensitivity model holds at $\Gamma$, $E(\sqrt{n}\bar{L}_\Gamma)\leq 0$, and $\var(\sqrt{n}\bar{L}_\Gamma) \leq nE\{\text{se}(\bar{L}_\Gamma)^2\}$, which tends to a limit as $n\rightarrow \infty$. That is, the unknown expectation is upper bounded by zero, and there exists a sample-based standard error whose expectation upper bounds the unknown variance.

\item[(iii)] The estimator $\text{se}(\bar{L}_\Gamma)^2$ converges in probability to its expectation, which is never smaller than $\text{var}(\bar{L}_\Gamma)$. For any $\epsilon >0$, we have that $\lim_{n\rightarrow\infty} \P(n\times \text{se}(\bar{L}_\Gamma)^2 + \epsilon \leq \var(\sqrt{n}\bar{L}_\Gamma)\} \rightarrow 0$ Fixing $\alpha \leq 0.5$ and sending $\epsilon$ to zero, 
\begin{align*}
\underset{n\rightarrow\infty}{\lim}\P\left\{\bar{L}_\Gamma \geq \text{se}(\bar{L}_\Gamma) \Phi^{-1}(1-\alpha)\right\} \leq \underset{n\rightarrow\infty}{\lim}\P\left\{\bar{L}_\Gamma \geq \text{sd}(\bar{L}_\Gamma) \Phi^{-1}(1-\alpha)\right\} \leq \alpha.
\end{align*}
\end{enumerate}

The second component of the proof involves showing that $\Phi^{-1}(1-\alpha)$ may be replaced by $G^{-1}_\Gamma(1-\alpha; \bm{1_n})$, the reference distribution generated by Algorithm 1 in the main text.  Showing this uses a technique from \citet{hoe52} for assessing the limiting behavior of a permutation distribution, along with the variant of Slutsky's theorem for randomization distributions from \citet{chu13}. To generate a reference distribution, we use $B_{\Gamma i}$ that stochastically dominates $L_{\Gamma i}$ under $H_{prop}^{(\lambda_0)}$, of the form 
\begin{align*}B_{\Gamma i} = V_{\Gamma i} |\zeta_i^{(\lambda_0)}| - \left(\frac{\Gamma -1}{\Gamma+1} \right) |\zeta_i^{(\lambda_0)}|,
\end{align*} 
where $V_{\Gamma i}=\pm 1$ are independent across pairs, with $\P(V_{\Gamma i} = 1) = \theta_\Gamma$. Let $(\sqrt{n}\bar{B}_\Gamma, \sqrt{n}\bar{B}'_\Gamma)$ be averages of $B_{\Gamma i}$ and $B'_{\Gamma i}$ formed with iid vectors $(\bm{V_\Gamma}$, $\bm{V'_\Gamma})$, such that $(\sqrt{n}\bar{B}_\Gamma, \sqrt{n}\bar{B}'_\Gamma)$ are identically distributed with covariance zero. We have
\begin{align*}E(B_{\Gamma i}) &= 0\\
\var(B_{\Gamma i}) &= E\{\var(B_{\Gamma i}\mid |\zeta_i^{(\lambda_0)}|\} + \var\{E(B_{\Gamma i}\mid |\zeta_i^{(\lambda_0)}|\}\\
&= 4\theta_\Gamma(1-\theta_\Gamma)\{\pi_i(\bar{\tau}_i + \eta_i)^2 + (1-\pi_i)(\bar{\tau}_i - \eta_i)^2\},
\end{align*} where $\theta_\Gamma, \pi_i, \eta_i$, and $\tau_i$ are defined as in Section \ref{sec:A1}. Further, under Conditions \ref{eqn:condi1} and \ref{eqn:condi2},  $(\sqrt{n}\bar{B}_\Gamma, \sqrt{n}\bar{B}'_\Gamma)^T$ tends in distribution to a multivariate normal, with mean vector zero, correlation zero, and equal variances $\nu_{\Gamma}^2 = 4 \theta_\Gamma (1-\theta_\Gamma) \nu^2$ where $\nu^2$ is defined in Condition~\ref{eqn:condi2}. We also have that $\sqrt{n}\text{se}(\bar{B}_\Gamma)$ converges in probability to $\nu_\Gamma$, such that $(\bar{B}_\Gamma/\text{se}(\bar{B}_\Gamma), \bar{B}'_\Gamma/\text{se}(\bar{B}'_\Gamma))^T$ converge in distribution to iid standard Normals. Recall that $G_\Gamma(\cdot; \bm{1}_n)$ is the CDF of $\bar{B}_\Gamma/\text{se}(\bar{B}_\Gamma)$. Combining the result of \citet{hoe52}, also given in Theorem 15.2.2 of \citet{leh05}, with Slutsky's theorem for randomization distributions \citep{chu13} yields that $\Phi^{-1}(1-\alpha)$ may be replaced by $G^{-1}_\Gamma(1-\alpha; \bm{1}_n)$ while preserving the asymptotic level of the procedure, proving the result.

\subsection{Proposition 2}
Let $\mathcal{D} = \{i: Y_{i1}+Y_{i2} = 1\}$, $\mathcal{C}_0 = \{i:Y_{i1}+Y_{i2}=0\}$, $\mathcal{C}_1 = \{i: Y_{i1}+Y_{i2} = 2\}$, and $\theta_\Gamma = \Gamma/(1+\Gamma)$. Under the sharp null, for any $\bz\in \Omega$, McNemar's test statistic can be written as $T_{\mathcal{D}}$ + $|\mathcal{C}_1|$, where $T_{\mathcal{D}} = T_{\mathcal{D}}(\bz, \bY) = \sum_{i\in \mathcal{D}}\sum_{j=1}^2z_{ij}Y_{ij}$. 

We now show that the studentized test statistic $\bar{L}_\Gamma/\text{se}(\bar{L}_\Gamma)$ is a monotone increasing function of $T_{\mathcal{D}}$ over $\bz \in \Omega$ under the sharp null. First observe that under the sharp null, for any $\bz \in \Omega$
\begin{align*}
n\bar{L}_\Gamma &= \sum_{i=1}^n\left\{(z_{i1}-z_{i2})(Y_{i1}-Y_{i2}) - (2\theta_\Gamma-1)|Y_{i1}-Y_{i2}|\right\}\\
&= \sum_{i \in \cD}(z_{i1}-z_{i2})(Y_{i1}-Y_{i2}) - (2\theta_\Gamma-1)|\cD|\\
&= \sum_{i \in \cD}\left\{z_{i1}Y_{i1} + z_{i2}Y_{i2} - z_{i1}(1-Y_{i1}) - z_{i2}(1-Y_{i2})\right\} - (2\theta_\Gamma-1)|\cD|\\
&= 2T_\cD + |\cD| - (2\theta_\Gamma-1)|\cD| = 2(T_\cD - \theta_\Gamma|\cD|),
\end{align*}

Meanwhile, $\sum_{i=1}^nL_{\Gamma i}^2$ can be expressed as
\begin{align*}
\sum_{i=1}^nL_{\Gamma i}^2 &= \{1+(2\theta_\Gamma-1)^2\}\sum_{i=1}^n(Y_{i1}-Y_{i2})^2 - 2(2\theta_\Gamma-1)\sum_{i=1}^n(Z_{i1}-Z_{i2})(Y_{i1}-Y_{i2})|Y_{i1}-Y_{i2}|\\
&= \{1+(2\theta_\Gamma-1)^2\}|\cD| - 2(2\theta_\Gamma-1)(2T_\cD + |\cD|)\\
&= \{1-(2\theta_\Gamma-1)^2\}|\cD| - 2(2\theta_\Gamma-1)\{2(T_\cD - \theta_\Gamma|\cD|)\},
\end{align*}
and $n(n-1)\text{se}(\bar{L}_\Gamma)^2$ as
\begin{align*}
n(n-1)\text{se}(\bar{L}_\Gamma)^2 = \{1-(2\theta_\Gamma-1)^2\}|\cD| - 2(2\theta_\Gamma-1)\{2(T_\cD - \theta_\Gamma|\cD|)\} - n^{-1}\{2(T_\cD - \theta_\Gamma|\cD|)\}^2.
\end{align*}
Let $a = n^{-1}$, $b=2(2\theta_\Gamma-1)$, $c=\{1-(2\theta_\Gamma-1)^2\}|\cD|$, and $x = 2(T_\cD - \theta_\Gamma|\cD|)$. We then have
\begin{align*}
\frac{\bar{L}_\Gamma}{\text{se}(\bar{L}_\Gamma)}\propto \frac{x}{\sqrt{c - bx - ax^2}},
\end{align*}
and to complete the proof it is sufficient to show that this function has a positive first derivative with respect to $x$ over its domain, $2\theta_\Gamma |\cD|\leq x \leq (2+ 2\theta_\Gamma)|\cD|$.

Differentiating yields
\begin{align*}
\frac{\partial}{\partial x}\frac{\bar{L}_\Gamma}{\text{se}(\bar{L}_\Gamma)} \propto \frac{\sqrt{c-bx - ax^2} + \frac{bx/2 + ax^2}{\sqrt{c-bx-ax^2}}}{c-bx-ax^2},
\end{align*}
and we want to show that this is always positive, which amounts to showing $c \geq (b/2)x$. Maximizing the right hand side over $x$, this means showing $\{1-(2\theta_\Gamma-1)^2\}|\cD| \geq (2\theta_\Gamma-1)(2|\cD| - 2\theta_\Gamma|\cD|)$, which does hold. $\bar{L}_\Gamma/\text{se}(\bar{L}_\Gamma)$ is monotone increasing in $2(T_\cD - \theta_\Gamma|\cD|)$, and hence also in $T_\cD + |\cC_1|$, which is McNemar's test. The test statistics have perfect rank correlation over $\bz \in \Omega$, and hence are equivalent.

\subsection{Proposition 3}

Recall that $\tilde{L}_{\Gamma i} = L_{\Gamma i}/\sqrt{1-h_{ii}}$, $\tilde{\mu}_{\Gamma i} = E(\tilde{L}_{\Gamma i})$ and $\tilde{\bm{\mu}}_{\Gamma} = E(\tilde{\mathbf{L}}_\Gamma)$. Define  $\bm\Lambda$ as the covariance matrix for $\tilde{L}_{\Gamma i} $, and note that $\bm{\Lambda}$ is a diagonal matrix, and the $i$-th diagonal element of $\bm\Lambda$ is $\Lambda_ii = \var(\tilde{L}_{\Gamma i}) = \var(L_{\Gamma i})/(1-h_{ii})$

Using results on expectations of quadratic forms, $E\{\text{se}(\bar{L}_\Gamma; \bQ)^2\}$ can be computed as
\begin{align*}
	E\{\text{se}(\bar{L}_\Gamma; \bQ)^2 \mid \mathcal{F}, \mathcal{Z}\} &= \frac{1}{n^2} E \left( \mathbf{\tilde{L}_{\Gamma}}^T (\mathbf{I} -\mathbf{H}) \mathbf{\tilde{L}_{\Gamma}} \right) \\
	&= \frac{1}{n^2} \left[ \tr\{(\bI - \bH) \bm{\Lambda}\} + \bm{\tilde{{\mu}}_\Gamma}^T (\bI - \bH) \bm{\tilde{{\mu}}_\Gamma}  \right] \\
	&= \frac{1}{n^2} \sum_{i=1}^{n} (1-h_{ii}) \Lambda_{ii} + \frac{1}{n^2}  \bm{\tilde{{\mu}}_\Gamma}^T (\bI - \bH) \bm{\tilde{{\mu}}_\Gamma} \\
	&= \frac{1}{n^2} \sum_{i=1}^{n}  \var({L}_{\Gamma i} \mid \mathcal{F} , \mathcal{Z}) + \frac{1}{n^2}  \bm{\tilde{{\mu}}_\Gamma}^T (\bI - \bH) \bm{\tilde{{\mu}}_\Gamma} \\
	&=  \var(\bar{L}_{\Gamma}) + \frac{1}{n^2}  \bm{\tilde{{\mu}}_\Gamma}^T (\bI - \bH) \bm{\tilde{{\mu}}_\Gamma} \geq \var(\bar{L}_{\Gamma}) 
\end{align*}
The last inequality holds because the matrix $(\bI - \bH)$ is positive semidefinite. 

\subsection{Proposition 4}

Define $\mathbf{H}_{\mathbf{Q}_{reg}} = \mathbf{Q}_{reg}(\mathbf{Q}_{reg}^T \mathbf{Q}_{reg})^{-1} \mathbf{Q}_{reg}^{T}$ and $\mathbf{H}_{\mathbf{1}_n} = \mathbf{1}_n(\mathbf{1}_n^T\mathbf{1}_n)^{-1} \mathbf{1}_n$. We prove the result for a general matrix $\bQ$ with a fixed number of columns for all $n$. Specifically, under suitable regularity conditions will prove
$$
n \times \se(\bar{L}_{\Gamma}; \bQ)^2 - \var(\sqrt{n} \bar{L}_{\Gamma}) \overset{p}{\to} \limn \frac{1}{n} {\bm\mu_{\Gamma}^T} (\bI - \bHQ) \bm\mu_{\Gamma}.
$$
The proof follows that of Theorem 1 in \citet{Fogarty:2018}. The following regularity conditions are considered to prove Proposition 4. 

\begin{condition} 
	
	Assume that $\bQ$ is a $(n \times k)$ matrix with $k$ fixed. The following conditions hold for each $\Gamma \geq 1$.
	
	\begin{itemize}
		\item[(1)]  (Bounded Fourth Moments).
		
		There exists a $C_1 < \infty$ such that $\frac{1}{n} \sum_{i=1}^{n} \E(L_{\Gamma i}^4) < C_1$ and $\frac{1}{n} \sum_{i=1}^{n} Q_{iq}^4 < C_1$ for $q=1, \ldots, k$. For the specification of $\bQ_{reg} = (\mathbf{1}, \bar{\mathbf{X}})$, the term $Q_{iq}^4$ can be replaced by $\bar{x}_{iq}^4$.
		
		\item[(2)] (Existence of Population Moments). 
		
		\begin{itemize}
			\item  $\frac{1}{n} \sum_{i=1}^{n} \mu_{\Gamma i}$, $\frac{1}{n} \sum_{i=1}^{n} \mu_{\Gamma i}^2$, and $\var(\sqrt{n} \bar{L}_{\Gamma}) = \frac{1}{n} \sum_{i=1}^{n} \var(L_{\Gamma i})$ converge to finite limits as $n \to \infty$. 
			
			\item $\frac{1}{n} \sum_{i=1}^{n} \mu_{\Gamma i} Q_{iq}$ converges to a finite limit for $q=1, \ldots, k$ as $n \to \infty$.
			
			\item $\frac{1}{n} \bQ^T \bQ$ converges to a finite, invertible matrix as $n \to \infty$. 
			
		\end{itemize}

	\end{itemize}
\label{eqn:condi3}
\end{condition}

Define $\kappa_{\Gamma q} = \underset{n \to \infty}{\lim} \frac{1}{n} \sum_{i=1}^{n} \mu_{\Gamma i} Q_{iq}$. This is the limit of the sum of all elements in the $q$th column of $\bm{\mu}_{\Gamma}^T \bQ$. Let $\bm\kappa_{\Gamma}$ as the vector of length $k$ containing limits $\kappa_{\Gamma q}$, $q=1, \ldots, k$. Also, define $\Sigma_{\bQ} = \underset{n \to \infty}{\lim} n^{-1} \bQ^T \bQ$. 

Consider the following lemma:
\begin{lemma}
	Under Condition~\ref{eqn:condi3},
	\begin{enumerate}
		\item $n^{-1} \sum_{i=1}^{n} L_{\Gamma i} Q_{iq}$ converges in probability to $\limn n^{-1} \sum_{i=1}^{n} \mu_{\Gamma i} Q_{iq}$. 
		
		\item $h_{ii} \to 0$
		
		\item $n^{-1} \sum_{i=1}^{n} L_{\Gamma i}^2 \overset{p}{\to} \limn n^{-1} 
		\sum_{i=1}^{n} \{\mu_{\Gamma i}^2 + \var(L_{\Gamma i})\}$
		
	\end{enumerate}
\label{lem1}
\end{lemma}

\begin{proof}[Proof of Lemma 1]
	\hfill\\
	Proof of (1). Since $\E(L_{\Gamma i}) = \mu_{\Gamma i}$, $\E(n^{-1} \sum_{i=1}^{n} L_{\Gamma i} Q_{iq}) \rightarrow \limn n^{-1} \sum_{i=1}^{n} \mu_{\Gamma i} Q_{iq}$. The variance $\var(n^{-1} \sum_{i=1}^{n} L_{\Gamma i} Q_{iq})$ goes to zero, 
	\begin{align*}
		\var(n^{-1} \sum_{i=1}^{n} L_{\Gamma i} Q_{iq}) &= \frac{1}{n^2} \sum_{i=1}^{n} Q_{iq}^2\var(L_{\Gamma i})\\
		& \leq  \frac{1}{n^2} \left\{\sum_{i=1}^{n} \var(L_{\Gamma i})^2 \right\}^{1/2} \left\{ \sum_{i=1}^{n} (Q_{iq})^4 \right\}^{1/2} \\
		& \leq  \frac{1}{n} \left\{ \frac{1}{n} \sum_{i=1}^{n} \E(L_{\Gamma i}^4) \right\}^{1/2} \left\{ \frac{1}{n} \sum_{i=1}^{n} (Q_{iq})^4 \right\}^{1/2} \to 0
	\end{align*}
	
\noindent Proof of (2). Since $h_{ii}$ can be represented by $Q_{i}^T (\bQ^T \bQ)^{-1} Q_i$ and $Q_i$ does not vary with $n$, from $\Sigma_{\bQ} = \underset{n \to \infty}{\lim} n^{-1} \bQ^T \bQ$, we have 
	$$
	\limn h_{ii} = \limn n^{-1} Q_i^T \Sigma_{\bQ}^{-1} Q_i = 0.
	$$

\noindent Proof of (3). This is a straightforward application of Chebyshev's inequality. Observe that $\var(n^{-1} \sum_{i=1}^{n} L_{\Gamma i}^2) \leq n^{-2} \sum_{i=1}^{n} \E(L_{\Gamma i}^4) < n^{-1} C_1\rightarrow 0$, such that average $n^{-1}\sum_{i=1}^nL^2_{\Gamma i}$ converges in probability to its expectation.
	
\end{proof}

For any $\bQ$, $n \times \se(\bar{L}_{\Gamma}; \bQ)^2$ is 
\begin{align*}
	n \times \se(\bar{L}_{\Gamma}; \bQ)^2 &= \frac{1}{n} \left( \tilde{\mathbf{L}}_{\Gamma}^T (\bI - \bHQ) \tilde{\mathbf{L}}_{\Gamma} \right) \\
	&=  \frac{1}{n} \left( \tilde{\mathbf{L}}_{\Gamma}^T \tilde{\mathbf{L}}_{\Gamma} -   \tilde{\mathbf{L}}_{\Gamma}^T \bQ (\bQ^T \bQ)^{-1} \bQ^{T}  \tilde{\mathbf{L}}_{\Gamma} \right).
\end{align*}
From Lemma~\ref{lem1}, the above expression converges in probability to 
\begin{align*}
	\frac{1}{n} \tilde{\mathbf{L}}_{\Gamma}^T \tilde{\mathbf{L}}_{\Gamma} &\overset{p}{\rightarrow} \limn \frac{1}{n} \bm{\mu}_{\Gamma}^T \bm\mu_{\Gamma} + \limn \var(\sqrt{n} \bar{L}_{\Gamma}) \\
	\frac{1}{n} \tilde{\mathbf{L}}_{\Gamma}^T \bQ (\bQ^T \bQ)^{-1} \bQ^{T}  \tilde{\mathbf{L}}_{\Gamma} &\overset{p}{\rightarrow} \bm\kappa_{\Gamma}^T \Sigma_{\bQ}^{-1} \bm\kappa_{\Gamma} = \limn \frac{1}{n} \bm\mu_{\Gamma}^T \bHQ \bm\mu_{\Gamma}.
\end{align*}
Therefore, we have 
$$
n \times \se(\bar{L}_{\Gamma}; \bQ)^2 - \var(\sqrt{n} \bar{L}_{\Gamma}) \overset{p}{\rightarrow}  \limn \frac{1}{n} \bm{\mu}_{\Gamma}^T \bm\mu_{\Gamma} -  \limn \frac{1}{n} \bm\mu_{\Gamma}^T \bHQ \bm\mu_{\Gamma} =  \limn \frac{1}{n} \bm\mu_{\Gamma}^T (\bI -\bHQ) \bm\mu_{\Gamma}.
$$

Also, since the coefficient $R^2$ in a regression of $\bm{\mu}_\Gamma$ on $\bQ_{reg}$ has the representation
$$
R^2 = 1 - \frac{\bm\mu_{\Gamma}^T (1- \bH_{\mathbf{Q}_{reg}}) \bm\mu_{\Gamma}}{\bm\mu_{\Gamma}^T (1- \bH_{\mathbf{1}_{n}})\bm\mu_{\Gamma} },
$$
we have
$$
\frac{n \times \se(\bar{L}_{\Gamma}; \bQ_{reg})^2 - \var(\sqrt{n} \bar{L}_{\Gamma})}{n \times \se(\bar{L}_{\Gamma}; {\mathbf{1}_n})^2 - \var(\sqrt{n} \bar{L}_{\Gamma})} \overset{p}{\rightarrow} 1- R^2. 
$$

\subsection{Proposition 5}

Assume that the pair index $i$ is re-ordered according to the pairs of pairs. For instance, for the $q$th of $n/2$ pairs of pairs, we assume that $i=2q-1$ and $i=2q$, $q=1, \ldots, n/2$, are in the $q$th pair of pairs. Then, the $n \times (n/2)$ matrix $\bQ_{PoP}$ and  the $n \times n$ matrix $\bHQ_{PoP}$ are represented by
$$
\bQ_{PoP} = \begin{bmatrix}
1 & & & \\
1 & & & \\
 & 1 & & \\
 & 1 & & \\
 & & \vdots & \\
 & & & 1 \\
 & & & 1 \\
\end{bmatrix}, \quad \bHQ_{PoP} = \begin{bmatrix}
1/2 & 1/2 & & & & & \\
1/2 & 1/2 & & & & & \\
& & 1/2 & 1/2 & & &\\
& & 1/2 & 1/2 & & &\\
& & & &\ddots & &\\
& & & & & 1/2 & 1/2 \\
& & & & & 1/2 & 1/2  \\
\end{bmatrix}.
$$

Using this representation of $\bQ$, the vector $(\bI - \bHQ_{PoP}) \tilde{\mathbf{L}}_{\Gamma}$ contains residuals. The $i$th component of this vector is $(L_{\Gamma i} - L_{\Gamma \mathcal{J}(i)})/\sqrt{2} $, where $\mathcal{J}(i)$ is index of the pair that was matched with to $i$th pair such that $\mathcal{J}(\mathcal{J}(i)) = i$. Since $(\bI - \bHQ_{PoP})$ is idempotent, the variance estimator $\se(\bar{L}_{\Gamma}; \bQ_{PoP})^2$ can be represented by the equation (5) in the main manuscript. Also, from Proposition 3, the bias $\E[\se(\bar{L}_{\Gamma}; \bQ_{PoP})^2 \mid \mathcal{F}, \mathcal{Z}] - \var(\bar{L}_{\Gamma} \mid \mathcal{F}, \mathcal{Z})$ is
\begin{align*}
	\frac{1}{n^2} \tilde{\bm{\mu}}_{\Gamma}^T (\bI - \bHQ_{PoP}) \tilde{\bm{\mu}}_{\Gamma} & = \frac{1}{n^2} \left\{(\bI - \bHQ_{PoP}) \tilde{\bm{\mu}}_{\Gamma}\right\}^T (\bI - \bHQ_{PoP}) \tilde{\bm{\mu}}_{\Gamma} \\
	& = \frac{1}{2n^2} \sum_{i=1}^{n} (\mu_{\Gamma i} - \mu_{\Gamma \mathcal{J}(i)})^2.
\end{align*}
The last equality stems from the fact that the $i$th component of $(\bI - \bHQ_{PoP}) \tilde{\bm\mu}_{\Gamma}$ is $(\mu_{\Gamma i} - \mu_{\Gamma \mathcal{J}(i)})/\sqrt{2} $.

\subsection{Proposition 6}

Recall that, for some matrix $\bQ$, the new testing procedure is defined as
$$
\varphi_{\bQ}^{(\lambda_0)} (\alpha, \Gamma) =  \1\{\bar{L}_\Gamma/\text{se}(\bar{L}_\Gamma; \bQ) \geq G_\Gamma^{-1}(1-\alpha; \bQ)\}
$$
where $G_{\Gamma}( \cdot; \bQ)$ is the distribution of $\bar{B}_{\Gamma}/\se(\bar{B}_{\Gamma}; \bQ)$ given $|\bm{\zeta^{(\lambda_0)}}|$ and is itself random over $\bz\in \Omega$. 

As in the proof of Proposition 1, we first establish
\begin{align*}
	\underset{n\rightarrow\infty}{\lim}\P\{\bar{L}_\Gamma/\text{se}(\bar{L}_\Gamma; \bQ) \geq \Phi^{-1}(1-\alpha)\} \leq \alpha,
\end{align*}
This holds under Condition \ref{eqn:condi3} for $\bQ_{reg}$ (or, more generally, for matrices $\bQ_{reg}$ whose number of columns do not grow with $n$), as under those conditions the standard errors have limits in probability. For $\bQ_{PoP}$, observe that the number of columns is $n/2$, such that Proposition 3 does not apply. It is sufficient to show that for any $\epsilon > 0$,
\begin{align*}
\underset{n\rightarrow\infty}{\lim}\P\{n\times \text{se}(\bar{L}_\Gamma;\bQ_{PoP})^2 \leq  n\times \var(\bar{L}_\Gamma) - \epsilon\} = 0
\end{align*}
The estimator $n\times \text{se}(\bar{L}_\Gamma; \bQ_{PoP})^2$ has expectation
\begin{align*}
n\times E\{\text{se}(\bar{L}_\Gamma; \bQ_{PoP})^2\}&= n\times \var(\bar{L}_\Gamma) + \frac{1}{2n} \sum_{i=1}^{n} (\mu_{\Gamma i} - \mu_{\Gamma \mathcal{J}(i)})^2.
\end{align*}
The term $(2n)^{-1} \sum_{i=1}^{n} (\mu_{\Gamma i} - \mu_{\Gamma \mathcal{J}(i)})^2$ is nonnegative, and Conditions \ref{eqn:condi1}-\ref{eqn:condi2} imply that $n\times \var(\bar{L}_\Gamma)$ tends to a (positive) limit. We have
\begin{align*}
\var\{n\times \text{se}(\bar{L}_\Gamma;\bQ_{PoP})^2\} \leq 16 n^{-2}\sum_{i=1}^nE(L_{\Gamma i}^4),
\end{align*}
which tends to zero under part (1) of Condition \ref{eqn:condi3}. Through Chebyshev's inequality, we then have for any $\epsilon > 0$
\begin{align*}
&\P\left\{n\times\text{se}(\bar{L}_\Gamma;\bQ_{PoP})^2 - n\times \var(\bar{L}_\Gamma) \leq -\epsilon\right\}\\ & \leq \P\left\{n\times\text{se}(\bar{L}_\Gamma;\bQ_{PoP})^2 - n\times \var(\bar{L}_\Gamma) - (2n)^{-1} \sum_{i=1}^{n} (\mu_{\Gamma i} - \mu_{\Gamma \mathcal{J}(i)})^2\leq -\epsilon\right\}\\
& \leq \P\left\{\left|n\times\text{se}(\bar{L}_\Gamma;\bQ_{PoP})^2 - n\times \var(\bar{L}_\Gamma) - (2n)^{-1} \sum_{i=1}^{n} (\mu_{\Gamma i} - \mu_{\Gamma \mathcal{J}(i)})^2\right|\geq \epsilon\right\}\\
&\leq \var\{n\times \text{se}(\bar{L}_\Gamma;\bQ_{PoP})^2\}/\epsilon^2\rightarrow 0,
\end{align*}
establishing
\begin{align*}
	\underset{n\rightarrow\infty}{\lim}\P\{\bar{L}_\Gamma \geq \text{se}(\bar{L}_\Gamma ; \bQ) \times \Phi^{-1}(1-\alpha)\} \} \leq \underset{n\rightarrow\infty}{\lim}\P\{\bar{L}_\Gamma \geq \text{sd}(\bar{L}_\Gamma ) \times \Phi^{-1}(1-\alpha)\} \} \leq \alpha
\end{align*} with both $\bQ = \bQ_{reg}$ and $\bQ = \bQ_{PoP}$.

Next, we show that $\Phi^{-1}(1-\alpha)$ may be replaced with $G^{-1}_\Gamma(1-\alpha;\bQ)$ for either of these choices for $\bQ$. The proof differs from the analogous portion of Proposition 1 only in that we must instead establish that 
\begin{align*} n \times \se(\bar{B}_{\Gamma}; \bQ)^2 \overset{p}{\rightarrow} \nu_{\Gamma}^2\end{align*} for $\bQ\neq \bm{1_n}$. This can be proved by using $\E\{se(\bar{B}_{\Gamma}; \bQ)^2\} = \var(\bar{B}_{\Gamma})$ for \textit{any} $\bQ$ and that $\var\{n \times \se(\bar{B}_{\Gamma}; \bQ)^2\} \to 0$. Recall that $\se(\bar{B}_{\Gamma}; \bQ)^2 = \frac{1}{n^2} \tilde{\bB}_{\Gamma}^T(\bI - \bHQ) \tilde{\bB}_{\Gamma}$ where $\bar{B}_{\Gamma i} = \frac{B_{\Gamma i}}{\sqrt{1-h_{ii}}}$. Define $\Lambda_{\Gamma}$ be the covariance matrix of $\tilde{\mathbf{B}}_{\Gamma}$, and $\Lambda_{\Gamma ii}$ be the $i$th diagonal element of $\Lambda_{\Gamma}$. We have $\E(B_{\Gamma i}) = 0$ and $\Lambda_{\Gamma ii} = \frac{1}{1-h_{ii}} \var(B_{\Gamma_i})$.

\begin{lemma}
	$\E\{\se(\bar{B}_{\Gamma}; \bQ)^2\} = \var(\bar{B}_{\Gamma})$
\end{lemma}

\begin{proof}
	$\E\{\se(\bar{B}_{\Gamma}; \bQ)^2\}$ can be computed by
\begin{align*}
\E[\se(\bar{B}_{\Gamma}; \bQ)^2] &= \frac{1}{n^2} \E[\tilde{\bB}_{\Gamma}^T(\bI - \bHQ) \tilde{\bB}_{\Gamma}] \\
&= \frac{1}{n^2} \tr[(\bI - \bHQ) \Lambda_{\Gamma}] + \E(\tilde{\bB}_{\Gamma})^T (\bI - \bHQ) \E(\tilde{\bB}_{\Gamma}) \\
&= \frac{1}{n^2} \tr\{(\bI - \bHQ) \Lambda_{\Gamma}\}\\
&= \frac{1}{n^2} \sum_{i=1}^{n} \var(B_{\Gamma i}) = \var(\bar{B}_{\Gamma})
\end{align*}
\end{proof}
The second to last line follows from $E(B_{\Gamma i}) = 0$ for all $i$, as the random variable $B_{\Gamma i}$ behaves as though the proportional dose model holds under the worst-case assignment probabilities. For $\bQ_{reg}$, by using the fact that $h_{ii} \to 0$ and assuming $\frac{1}{n} \sum_{i=1}^{n} \E(B_{\Gamma i}^4)$ is bounded, $n\times \var\{\se(\bar{B}_{\Gamma}; \bQ_{reg})^2\} \to 0$ through an analogous argument to that in Proposition 4. Assuming $n^{-1}\sum_{i=1}^nE(\bar{B}_{\Gamma i})$ is bounded also gives that $n\times \var\{\se(\bar{B}_{\Gamma}; \bQ_{PoP})^2 \} \to 0$, such that $n \times \se(\bar{B}_{\Gamma}; \bQ_{PoP})^2 \overset{p}{\rightarrow} \nu_{\Gamma}^2$. With this established, the proof that the quantile $G^{-1}_\Gamma(1-\alpha ; \bQ)$ with $\bQ = \bQ_{reg}$ or $\bQ_{PoP}$,  may be employed is analogous to Proposition 1.

To prove that $\varphi_{\bQ}^{(\lambda_0)}(\alpha,\Gamma)$ with $\bQ = \bQ_{reg}$ or $\bQ = \bQ_{PoP}$ is both less conservative under the null and more powerful under the alternative than the test using the conventional standard error, it suffices to show that for any $\epsilon > 0$ and for either of these choices $\bQ$,
\begin{align*}
\underset{n\rightarrow\infty}{\lim}\P[n\times\left\{\text{se}(\bar{L}_\Gamma; \bm{1_n})^2 - \text{se}(\bar{L}_\Gamma ; \bQ)^2\right\} \leq -\epsilon]  &= 0
\end{align*}

Observe that
\begin{align*}
&\P[n\left\{\text{se}(\bar{L}_\Gamma; \bm{1_n})^2 - \text{se}(\bar{L}_\Gamma ; \bQ)^2\right\} \leq -\epsilon]\\
&=\P\left(n\left[\text{se}(\bar{L}_\Gamma; \bm{1_n})^2 - E\{\text{se}(\bar{L}_\Gamma; \bm{1_n})^2\} +  E\{\text{se}(\bar{L}_\Gamma; \bm{1_n})^2\}- \text{se}(\bar{L}_\Gamma ; \bQ)^2\right] \leq -\epsilon\right)\\
&\leq \P\left(n\left[\text{se}(\bar{L}_\Gamma; \bm{1_n})^2 - E\{\text{se}(\bar{L}_\Gamma; \bm{1_n})^2\}\right] \leq -\epsilon\right) +  \P\left(n\left[E\{\text{se}(\bar{L}_\Gamma; \bm{1_n})^2\}- \text{se}(\bar{L}_\Gamma ; \bQ)^2\right] \leq -\epsilon\right)\\
\end{align*}
The first probability tends to zero, as under Conditions \ref{eqn:condi1}-\ref{eqn:condi2} the conventional standard error converges in probability to its expectation. For $\bQ_{reg}$, the second probability also tends to zero under Condition \ref{eqn:condi3} through a proof analogous to Proposition 4. For $\bQ_{PoP}$, we need conditions to further ensure that
\begin{align*}
\underset{n\rightarrow\infty}{\lim}\left\{\frac{1}{n}\sum_{i=1}^n(\mu_{\Gamma i} - \bar{\mu}_\Gamma)^2 - \frac{1}{2n}\sum_{i=1}^n(\mu_{\Gamma i} - \mu_{\Gamma \mathcal{J}(i)})^2\right\} \geq 0.
\end{align*} 
In a sense, the required condition is that pairing pairs on the basis of observed covariates is effectively creating pairs of pairs with similar expectations. Observe that if we randomly paired the pairs without any consideration of similarity on observed covariates, the difference in these two terms would \textit{equal} zero. Regularity conditions are needed to preclude pairings of pairs yielding a worse alignment of expectations than what would be expected under random pairings of pairs. The sufficient condition reflecting this natural requirement takes the form
\begin{align*}
\underset{n\rightarrow\infty}{\lim}\frac{1}{2n}\sum_{i=1}^n(\mu_{\Gamma i} - \bar{\mu}_\Gamma)(\mu_{\Gamma \mathcal{J}(i)} - \bar{\mu}_\Gamma) \geq 0.
\end{align*}
See \citet{Abadie:2008} for a discussion of sufficient conditions in a superpopulation formulation where the potential outcomes are viewed as random but the covariates as fixed, wherein natural connections between our sufficient condition and Lipschitz conditional expectation and conditional variance functions are explored.

\section{Additional details on improved standard errors}

\subsection{Comparative improvement and fundamental limitations of bias reduction for finite-population standard errors}\label{sec:limit}
\citet{Abadie:2008} consider a superpopulation formulation wherein the covariates $\bx$ are viewed as fixed, but the potential outcomes are random. Under this generative framework, they show that under suitable regularity conditions, $n\times \text{se}(\bar{L}_1; \bQ_{PoP})^2$ is consistent for $n\times\var(\bar{L}_1\mid \bx, \cZ)$ in randomized experiments with  perfect compliance. Note however that even at $\Gamma=1$, the target of estimation in our framework is instead  $\var(\bar{L}_1\mid \cF, \cZ)$, and that $\cF$ includes the potential outcomes. That inference conditions upon the potential outcomes results in what \citet{Ding:2019} refer to as an \textit{idiosyncratic} component of heterogeneity which cannot be explained by observed covariates. As a result, $n\times \text{se}(\bar{L}_\Gamma; \bQ_{PoP})^2$ will not generally be consistent for $n\times\var(\bar{L}_\Gamma \mid \cF, \cZ)$. Instead, it will have a positive asymptotic bias determined by the limiting value of $(2n)^{-1}\sum_{i=1}^n(\mu_{\Gamma i} - \mu_{\Gamma \mathcal{J}(i)})^2$ given $\cF$ and $\cZ$, which generally cannot be driven to zero as it depends on idiosyncratic features of the observed study population. Nonetheless, by pairing the pairs one may reasonably expect that the explainable portion of the variation in $\mu_{\Gamma i}$ has been captured in the limit so long as the underlying function is sufficiently well-behaved, whereas the method based on linear regression may fail to entirely do so under misspecification.

To develop intuition, imagine for a moment that the terms $\mu_{\Gamma i} = E(\bar{L}_{\Gamma i}\mid \cF, \cZ)$ are themselves drawn from a distribution conditional upon $\bx_i$ with expectation $\eta(\bx_i)$ and variance $\sigma^2(\bx_i)$. Let $\varepsilon(\bx_i) = \mu_{\Gamma i} - \eta(\bx_i)$ such that $E\{\varepsilon(\bx_i) \mid \bx_i\} = 0$. Under suitable regularity conditions,  the limiting bias in $n\times \text{se}(\bar{L}_\Gamma; \bQ_{reg})^2$ as an estimator of $ n\times \var(\bar{L}_\Gamma\mid \cF, \cZ)$ is
\begin{align*}n\{\text{se}(\bar{L}_\Gamma;\bQ_{reg})^2 - \var(\bar{L}_\Gamma\mid \cF, \cZ)\} &\overset{p}{\rightarrow} \underset{n\rightarrow \infty}{\lim}n^{-1} \{\bm{\eta}(\bx) + \bm{\varepsilon}(\bx)\}^T(\bI - \bH_{reg})\{\bm{\eta}(\bx) + \bm{\varepsilon}(\bx)\}
	\\&= n^{-1} \underset{n\rightarrow \infty}{\lim}\bm{\eta}(\bx)^T(\bI - \bH_{reg})\bm{\eta}(\bx) + n^{-1}\underset{n\rightarrow \infty}{\lim}\sum_{i=1}^n\varepsilon^2(\bx_i).\end{align*}
The term $n^{-1}\underset{n\rightarrow \infty}{\lim}\sum_{i=1}^n\varepsilon^2(\bx_i)$ in the expression is the idiosyncratic effect variation described in \citet{Ding:2019}, and appears because inference is being conducted conditional upon $\cF$ with the potential outcomes fixed, rather than conditional upon only $\bx_i$ while viewing the potential outcomes as random. The bias would only depend upon this idiosyncratic term were $\eta(\bx_i)$ linear in $\bx_i$, but an additional positive factor persists under misspecification. For the pairs of pairs estimator, under suitable moment conditions given in \citet{Abadie:2008},  we instead have
\begin{align*}n\{\text{se}(\bar{L}_\Gamma;\bQ_{PoP})^2 - \var(\bar{L}_\Gamma\mid \cF, \cZ)\} &\overset{p}{\rightarrow} \underset{n\rightarrow \infty}{\lim}(2n^{-1}) \sum_{i=1}^n\{\eta(\bx_i) - \eta(\bx_{\mathcal{J}(i)}) + \varepsilon(\bx_i) - \varepsilon(\bx_{\mathcal{J}(i)})\}^2\\
	&= \underset{n\rightarrow \infty}{\lim}(2n^{-1}) \sum_{i=1}^n\{\eta(\bx_i) - \eta(\bx_{\mathcal{J}(i)})\}^2 + n^{-1}\underset{n\rightarrow \infty}{\lim}\sum_{i=1}^n\varepsilon^2(\bx_i),\end{align*}
Under additional Lipschitz and boundedness conditions on $\eta(\bx_i)$ along with an assumption that $n^{-1}\sum_{i=1}^n||\bx_i - \bx_{\mathcal{J}(i)}||^2$ tends to zero, arguments akin to those in \citet{Abadie:2008} imply that $\underset{n\rightarrow \infty}{\lim}(2n)^{-1}\sum_{i=1}^n\{\eta(\bx_i) - \eta(\bx_{\mathcal{J}(i)})\}^2 = 0$. Under these conditions, it would then follow that $\underset{n\rightarrow \infty}{\text{plim}}\; \text{se}(\bar{L}_\Gamma;\bQ_{PoP})/\text{se}(\bar{L}_\Gamma;\bQ_{reg}) \leq 1$, such that $\text{se}(\bar{L}_\Gamma ; \bQ_{PoP})$ would provide a less conservative standard error in the limit than $\text{se}(\bar{L}_\Gamma;\bQ_{reg})$.
\subsection{Pairs of pairs with an odd number of pairs}
An odd number of pairs can readily be accommodated through a host of corrections. One approach is inspired by the output of the \texttt{nonbimatch} function within the \texttt{nbpMatching} package in \texttt{R} \citep{Lu:2011}, a common package for nonbipartite matching. When given an odd number of elements, \texttt{nonbimatch} introduces a ``ghost" element has zero distance between all other elements, and then proceeds with the optimization problem. The pair matched to the ghost is, in reality, not matched to any other pair. We can then match the unmatched pair to the most similar pair of pairs, creating $(n-1)/2-1$ pairs of pairs and one triple of pairs. One could then proceed with $\bQ_{PoP}$ as the $n\times (n-1)/2$ matrix, where the first $(n-1)/2-1$ columns contain indicators for pairs of pairs membership, and the final column contains binary indicators for membership in the lone triple.  The estimator $\text{se}(\bar{L}_\Gamma ; \bQ_{PoP})^2$ is a conservative estimate for $\var(\bar{L}_\Gamma \mid \cF, \cZ)$. 
\newpage

\subsection{Pseudocode for conducting the sensitivity analysis}
The following algorithm outlines how one would perform a sensitivity analysis using the improved standard errors:
\begin{algorithm}[h]
\caption{Studentized sensitivity analysis at $\Gamma$ with improved standard errors\label{alg:stu}}
\begin{enumerate} 
\item Form a matrix $\bQ$ and its hat matrix $\bH = \bQ(\bQ^T\bQ)^{-1}\bQ^T$
\item In the $m$th of $M$ iterations:
\begin{enumerate}
\item Generate $V_{\Gamma i}\overset{iid}{\sim} 2\times Bernoulli\left(\frac{\Gamma}{1+\Gamma}\right)-1$ for each $i$.
\item Compute $B_{\Gamma i} = V_{\Gamma i}|\zeta^{(\lambda_0)}_i| - \left(\frac{\Gamma-1}{1+\Gamma}\right)|\zeta^{(\lambda_0)}_i|$ for each $i$.
\item Compute $\bar{B}_\Gamma = n^{-1}\sum_{i=1}^nB_{\Gamma i}$.
\item Form $\tilde{B}_{\Gamma i} = B_{\Gamma i}/\sqrt{1-h_{ii}}$ for each $i$
\item Compute $\text{se}(\bar{B}_\Gamma;\bQ)$ as 
\begin{align*}
\text{se}(\bar{B}_\Gamma;\bQ) = \sqrt{\frac{1}{n^2}\mathbf{\tilde{B}_\Gamma}^T(\mathbf{I}-\bH)\mathbf{\tilde{B}_\Gamma}}.
 \end{align*}
\item Compute $A^{(m)}_\Gamma = \bar{B}_\Gamma/\text{se}(\bar{B}_\Gamma; \bQ)$; store this value across iterations. 
\end{enumerate}
\item Approximate the bound on the greater-than $p$-value by \begin{align*} \hat{p}_{val} &= \frac{1 + \sum_{m=1}^M\1\{A^{(m)}_\Gamma\geq \bar{L}^{obs}_\Gamma/\text{se}(\bar{L}^{obs}_\Gamma ; \bQ)\}}{1+M}\end{align*}
\end{enumerate}
\end{algorithm}

\section{Simulation studies highlighting the improved standard errors}
\subsection{A simulation with no hidden bias}

We now conduct a simulation study to further illustrate the potential improvements that standard errors exploiting effect modification may provide. In each simulated data set there are $n$ pairs, each matched exactly on a $k\geq 5$ dimensional vector of covariates $\bx_i$. In the $m$th of $M$ iterations, the pair-specific covariate vector is drawn such that each component is $iid$ Uniform on the interval [0,1]. The treatments actually received $d_{ij}(z)$ ($z=0,1$) are binary. We assume that there are no defiers, that the exclusion restriction holds, and that individuals $ij$ are assigned status as compliers, never-takers, and always-takers independently with probability $p_C = 0.58$ (the estimated compliance rate from our data set), $p_N = 0.21$, and $p_C = 0.21$. We then use the functional form for the potential outcomes suggested in \citet[\S 7.1]{Fogarty:2018}, suitably modified to include potential noncompliance. For $z=0,1$ and with $\varepsilon_{ij}$ $iid$ standard Normal,\begin{align*} 
y_{ij}(z) &=\begin{cases} a\left(10\sin(\pi x_{i1}x_{i2}) + 20(x_{i3}-1/2)^2 + 10\exp(x_{i4}) +5(x_{i5}-1/2)^3 + \varepsilon_{ij}\right) & d_{ij}(z) = 1\\
 10\sin(\pi x_{i1}x_{i2}) + 20(x_{i3}-1/2)^2 + 10\exp(x_{i4}) + 5(x_{i5}-1/2)^3 + \varepsilon_{ij} & d_{ij}(z) = 0\nonumber \end{cases}
.\nonumber
\end{align*}
 We proceed with two different values of $a$ in the above model: $a=1$ and $a=2$. At $a=1$ the proportional dose model holds at $\lambda=0$, and there is no effect modification. At $a=2$ effect modification is present and is nonlinear in the observed covariates.


The effect ratio in simulation $m$ is $\lambda_m$, and is determined once the potential outcomes and dosages received are simulated. We imagine that there is no hidden bias, such that the sensitivity model holds at $\Gamma=1$. We then proceed with inference for the null hypothesis $\lambda = \lambda_m$ at $\Gamma=1$ with a two-sided alternative, and we set $\alpha=0.1$. We choose this larger value for $\alpha$ because of the conservativeness of inference in the presence of effect heterogeneity: in simulations with heterogeneous effects, smaller values of $\alpha$ may lead to near zero true Type I error rates.  We compare three choices for $\bQ$ for the variance estimator:\vspace{-.1 in}\begin{enumerate}
\item[(a)] \textit{Intercept}. The usual variance estimator for a paired design. $\mathbf{Q} = \bm{1}_n$.
\item[(b)] \textit{Linear.} Including a constant column and the covariates $\bar{\bx}_{i}$. $\mathbf{Q} = (\bm{1}_n, \mathbf{\bar{X}})$.
 
\item[(c)] \textit{Pairs of Pairs}. Form pairs of pairs with similar covariate values through nonbipartite matching as described in \S 5.4 of the manuscript. $\mathbf{Q} = \mathbf{Q}_{PoP}$.
\end{enumerate}We also construct 90\% confidence intervals for $\lambda_m$ through inverting $\varphi_\bQ^{(\lambda_0)}$, and compare the interval widths across  the three choices for $\bQ$. We vary the value of $a$, the sample size between $n=100$ and $n=2000$, and the number of covariates between $k=5$ and $k=10$. Note that the settings with $k=10$ include five irrelevant covariates.



\begin{table}[h]
\begin{center}
\caption{\label{tab:effmodsim} Simulated performance with different standard error estimators. The desired Type I error rate in all settings is $\alpha=0.1$.}
\begin{tabular}{l l l c c c c c c}
\toprule
&&&\multicolumn{2}{c}{Intercept}&\multicolumn{2}{c}{Linear} & \multicolumn{2}{c}{Pairs of Pairs}\\
&&& Size & CI Length & Size & CI Length & Size & CI Length\\
\midrule
Prop. Dose & $n=100$ & $k=5$ & 0.097& 0.840& 0.098& 0.840 &0.100 & 0.846\\
Prop. Dose & $n=100$& $k=10$& 0.103 & 0.847 &0.104 & 0.848 &0.103 & 0.853\\
Prop. Dose & $n=300$ & $k=5$ & 0.099 & 0.470 & 0.099 & 0.470 & 0.099 & 0.471\\
Prop. Dose & $n=300$& $k=10$ & 0.105& 0.472 &0.105 &0.472& 0.104&0.474\\
Prop. Dose & $n=1000$ & $k=5$ & 0.010& 0.256& 0.098& 0.256 &0.099 & 0.256\\
Prop. Dose & $n=1000$& $k=10$& 0.102 & 0.256 &0.102 & 0.256 &0.104 & 0.256 \\
Prop. Dose & $n=2000$ & $k=5$ & 0.101 & 0.181 & 0.101 & 0.181 & 0.102& 0.181\\
Prop. Dose & $n=2000$& $k=10$ & 0.102& 0.181 &0.102 &0.181& 0.101&0.181\\
\midrule
Effect Mod. & $n=100$& $k=5$ & 0.015& 3.241 & 0.043 &2.652 & 0.043 &2.730\\
Effect Mod. & $n=100$& $k=10$& 0.013 &3.240 & 0.040 & 2.657 & 0.029 & 2.873\\
Effect Mod.& $n=300$& $k=5$& 0.009 & 1.826 & 0.042 & 1.475 & 0.045 & 1.47\\
Effect Mod.& $n=300$& $k=10$& 0.011 & 1.825 & 0.042 & 1.475 &0.033 &1.555\\
Effect Mod. & $n=1000$& $k=5$ & 0.010& 0.996 & 0.042 &0.802 & 0.048 &0.780\\
Effect Mod. & $n=1000$& $k=10$& 0.012 &0.997& 0.040 & 0.802 & 0.036 & 0.824\\
Effect Mod.& $n=2000$& $k=5$& 0.011 & 0.701 & 0.040 & 0.567 & 0.048 & 0.546\\
Effect Mod.& $n=2000$& $k=10$& 0.011 & 0.701 & 0.042 & 0.567 &0.039&0.575\\
\bottomrule
\end{tabular}
\end{center}
\end{table}
Table \ref{tab:effmodsim} contains the results from the simulation study. We first observe that under the proportional dose model, the estimated size for all three procedures was roughly 0.1 in all settings. This is expected: under the proportional dose model all three procedures are guaranteed to be finite-sample exact since they are randomization tests. The confidence interval lengths are also comparable within each combination of $k$ and $n$ since all three standard errors times $\sqrt{n}$ are consistent for $\sqrt{n}\times \sd(\bar{L}_\Gamma \mid \cF, \cZ)$ under proportional doses.

With effect modification, in all settings the true Type I error rates from the resulting methods fall below the desired size.  The simulation study contains idiosyncratic, irreducible effect heterogeneity in the form of the $\varepsilon_{ij}$ terms, which induces conservativeness in the resulting inference. The regression-based and pairs of pairs standard errors both outperform the conventional standard errors, which results in less conservative inference and narrower confidence intervals. Despite the fact that the effect modification is nonlinear in the observed covariates, the regression-based standard error outperforms the pairs of pairs procedure in every simulation setting except for ($n=1000, k=5$) and ($n=2000, k=5)$. This is due to the difficulties faced by nonparametric estimation in low sample, high covariate dimension regimes. This does not corrupt the procedure based on the pairs of pairs standard error in terms of invalid inference, but it does limit its efficacy. This reflects the bias-variance tradeoff: in small samples and/or with a large number of covariates, it may be the case that a misspecified linear model produces better results than the nonparametric pairs of pairs approach due to improved stability of the estimated functional form.

\subsection{A simulation with hidden bias present}
Here we modify the simulation study in the previous section such that unmeasured confounding is actually present. We use the same generative model for the observed covariates and the potential outcomes, but rather than assuming no hidden bias we assume that the sensitivity model holds at $\Gamma=1.5$ and that nature presents us with the worst-case unmeasured confounders in each matched pair that lead to the largest expectation for the test statistic, making it such that the larger value of the two potential values for $\zeta_i^{(\lambda_m)} - \{(1.5-1)/(1+1.5)\}|\zeta_i^{(\lambda_m)}|$ occurs with probability 1.5/(1+1.5)=0.6 in each pair, and the smaller value with probability 0.4. We conduct a sensitivity analysis at $\Gamma=1.5$, such that the procedures should control the Type I error rate and provide sensitivity intervals with valid coverage. We again use the conventional standard error, the regression-based standard error, and the pairs of pairs standard error, allow for both homogeneous and heterogeneous effects, and conduct our simulation study with $n=100, 300, 1000, 2000$. For each of the 16 simulation settings, we simulate 1000 data sets and conduct a sensitivity analysis at $\Gamma=1.5$ with at two-sided alternative at $\alpha = 0.1$. We also compute 90\% sensitivity intervals.

\begin{table}[h]
\begin{center}
\caption{\label{tab:effmodsimgam} Simulated performance with different standard error estimators. The desired Type I error rate in all settings is $\alpha=0.1$. The sensitivity model holds at $\Gamma=1.5$, and we conduct a sensitivity analysis at $\Gamma=1.5$ with a two-sided alternative.}
\begin{tabular}{l l l c c c c c c}
\toprule
&&&\multicolumn{2}{c}{Intercept}&\multicolumn{2}{c}{Linear} & \multicolumn{2}{c}{Pairs of Pairs}\\
&&& Size & CI Length & Size & CI Length & Size & CI Length\\
\midrule
Prop. Dose & $n=100$ & $k=5$ & 0.045& 1.791& 0.050& 1.792 &0.045 & 1.796\\
Prop. Dose & $n=100$& $k=10$& 0.050 & 1.785 &0.053 & 1.785 &0.051 & 1.788\\
Prop. Dose & $n=300$ & $k=5$ & 0.046 & 1.320 & 0.045 & 1.320 & 0.046 & 1.322\\
Prop. Dose & $n=300$& $k=10$ & 0.051& 1.318 &0.049 &1.318& 0.050&1.318\\
Prop. Dose & $n=1000$ & $k=5$ & 0.056& 1.071& 0.057& 1.071 &0.058 & 1.071\\
Prop. Dose & $n=1000$& $k=10$& 0.046 & 1.072 &0.045 & 1.072 &0.044 & 1.073 \\
Prop. Dose & $n=2000$ & $k=5$ & 0.045 & 0.988 & 0.043 & 0.988 & 0.043& 0.987\\
Prop. Dose & $n=2000$& $k=10$ & 0.057& 0.988 &0.059 &0.988& 0.058&0.987\\
\midrule
Effect Mod. & $n=100$& $k=5$ & 0.001& 6.457 & 0.005 &5.886 & 0.004 &5.958\\
Effect Mod. & $n=100$& $k=10$& 0.001 &6.510 & 0.005 & 5.960 & 0.003 & 6.171\\
Effect Mod.& $n=300$& $k=5$& 0.000 & 4.875 & 0.000& 4.550 & 0.000 & 4.536\\
Effect Mod.& $n=300$& $k=10$& 0.000 & 4.866 & 0.000 &4.571  &0.000&4.593\\
Effect Mod. & $n=1000$& $k=5$ & 0.000& 3.940 & 0.000 &3.757 & 0.000 &3.732\\
Effect Mod. & $n=1000$& $k=10$& 0.000 &3.941& 0.000& 3.758 & 0.000 & 3.777\\
Effect Mod.& $n=2000$& $k=5$& 0.000& 3.631 & 0.000 & 3.503 & 0.000 & 3.483\\
Effect Mod.& $n=2000$& $k=10$& 0.000& 3.637 & 0.000 & 3.508 &0.000&3.514\\
\bottomrule
\end{tabular}
\end{center}
\end{table}
Table \ref{tab:effmodsimgam} contains the results from the simulation study. We first observe that under the proportional dose model, the estimated size for all three procedures was below 0.1 in all settings and fell closer to 0.05. That the size doesn't exceed 0.1 for any sample size is to be expected: the proportional dose model holds and all three procedures are are randomization tests. The confidence interval lengths are also comparable within each combination of $k$ and $n$ since all three standard errors times $\sqrt{n}$ are consistent for $\sqrt{n}\times \sd(\bar{L}_\Gamma \mid \cF, \cZ)$ under proportional doses. At $\Gamma=1$, we saw in the main text that the estimated size was exactly 0.1, rather than below 0.1, in all simulation settings. This conservativeness at $\Gamma=1.5$ has to do with our choice of a two-sided alternative, and would not have been present had we chosen a greater-than alternative instead. At $\Gamma=1$, there is a single expectation for the test statistic being deployed; however, in the sensitivity analysis, the researcher must consider a range of expectations for different patterns of hidden bias. To reject with a two-sided alternative, the test statistic must be either significantly above the upper-bound on the expectation, or significantly below the lower-bound on the expectation. At $\Gamma=1$ the upper and lower bounds are one in the same, but in a sensitivity analysis they are not. In our simulation the true pattern of bias was chosen to maximize the expectation of the test statistic. The test statistic falling significantly above the true expectation can then lead to a false rejection, but falling below the true expectation likely will not, as the test statistic will typically still be above the lower bound on the expectation. As rejections are by and large happening in the right tail only, the resulting Type I error rate falls closer to $\alpha/2 = 0.05$.

With effect modification, in all settings the true Type I error rates from the resulting methods fall well below the desired size, with most settings recording zero rejections in any data set. In a sensitivity analysis effect heterogeneity not only results in conservative standard errors, but it also frequently yields conservative bounds on the worst-case expectation \citep{Fogarty:2019} which unfortunately cannot be overcome without risking an anti-conservative procedure if effects are instead homogeneous. Looking at the lengths of the sensitivity intervals, we see that just as in the simulation assuming no hidden bias the regression-based and pairs of pairs standard errors both outperform the conventional standard errors, which results in less conservative inference and narrower confidence intervals. Despite the fact that the effect modification is nonlinear in the observed covariates, the regression-based standard error outperforms the pairs of pairs procedure in every simulation setting except for ($n=1000, k=5$) and ($n=2000, k=5)$. This is due once again to the difficulties faced by nonparametric estimation in low sample, high covariate dimension regimes.

\section{An omnibus test for effect heterogeneity in instrumental variable studies}
\subsection{Testing for effect modification}
\citet[\S 6]{Fogarty:2018} describes how the discrepancy between standard errors involving covariate information and the conventional standard error estimator for a paired design can be used to form an exact test for the null hypothesis of no effect modification. In the context of an instrumental variable design, this amounts to a test of the null hypothesis that the proportional dose model holds for some $\lambda_0$, against the alternative that it does not hold for all $i,j$ and for any shared value of $\lambda_0$.

Suppose that $\Gamma=1$, and for a given value $\lambda_0$ consider as a test statistic the $F$-statistic comparing a regression of $\zeta_i^{(\lambda_0)}$ on (a) a design matrix incorporating covariates, such as $\bQ_{reg}$ or $\bQ_{PoP}$; to (b) the null model containing only an intercept column. If effect modification existed on the basis of the observed covariates, one would expect the regression incorporating covariates to reduce the sum of squared error relative to a model containing only an intercept, which would in turn inflate the $F$-statistic. Call the resulting statistic $F(\bZ, \bm{\zeta}^{(\lambda_0)})$. If the proportional dose model holds at $\lambda_0$, then $|\zeta_i^{(\lambda_0)}|$ would be fixed across randomizations with only its sign varying. This would allow for computation of a randomization-based $p$-value for the proportional dose model holding at $\lambda_0$:
\begin{align}\label{eq:pval}
p^{(\lambda_0)} &= \frac{1}{2^n}\sum_{\bz \in \Omega}\1\{F(\bz, \bm{\zeta}^{(\lambda_0)}) \geq f^{(\lambda_0)}\},
\end{align}
where $f^{(\lambda_0)}$ is the observed value of the test statistic.

The value of $\lambda_0$ is a nuisance parameter for this test. Under perfect compliance, i.e. the typical paired experiment, Proposition 3 of \citet{Fogarty:2018} shows that $F(\bZ, \bm{\zeta}^{(\lambda_0)})$ is pivotal with respect to the particular value of $\lambda_0$, such all values for $\lambda_0$ yield the same $p$-value, and any value for $\lambda_0$ could be used. Unfortunately, this does not hold in the presence of noncompliance, meaning that $p^{(\lambda_0)}$ varies as a function of $\lambda_0$. To overcome this, we employ the approach of \citet{Berger:1994} to create a test of this hypothesis. By the lemma of \citet[\S 2]{Berger:1994}, a valid $p$-value for testing the null hypothesis that the proportional dose model holds is
\begin{align}\label{eq:pbeta}
p_\beta &=  \underset{\lambda_0 \in CI_{1-\beta}}{\sup}p^{(\lambda_0)} + \beta.
\end{align}
\noindent where $CI_{1-\beta}$ is \textit{any} $100(1-\beta)\%$ confidence interval for $\lambda_0$ under the assumption of the proportional dose model. This permits the use of randomization-based confidence intervals derived in \citet{Imbens:2005b} and \citet[\S 5.4]{Rosenbaum:2002} formed using test statistics other than the difference in means. However, for the purposes of improving the power of the test when effect modification is present, the confidence intervals described in \S 6 of the manuscript based on improved standard errors may be preferred. 
Through employing a randomization distribution, the resulting procedure yields an exact test for any sample size. See  \citet{ding2015randomization} for other examples of randomization-based tests of treatment effect heterogeneity.

An omnibus test of effect heterogeneity is of special interest for instrumental variable studies. Even if one assumes monotonicity and the exclusion restriction, the effect ratio merely attests to the sample average treatment effect among compliers, rather than the average treatment effect for the study population. Some authors argue that a key weakness of IV designs is that they only identify this more local estimand  \citep{deaton2010instruments,swanson2014think,swanson2017challenging}. For an IV design to provide an estimate of the SATE, the investigator must invoke an effect homogeneity or no-interaction assumption \citep{robins1994correcting, Hernan:2006}. If we fail to reject the null hypothesis, there is no evidence to reject the proportional dose model or to suggest that effect heterogeneity is present. If effects are roughly homogeneous, the IV estimand may be a good proxy for the sample average treatment effect in the study population.

\subsection{An application to our data set}
To formally test for effect modification within our data set, we use the test of the proportional dose model described in the previous subsection. We set $\beta=0.01$, and construct a $99\%$ confidence interval for $\lambda_0$ by inverting the test $\varphi^{(\lambda_0)}_{\bQ_{PoP}}$ at $\Gamma=1$. We then maximize (\ref{eq:pval}) over $\lambda_0$ in this confidence interval, and form $p_{0.01}$ through (\ref{eq:pbeta}). For septic patients, the resulting $p$-value for the complication outcome is 0.65, and the $p$-value is 0.10 for the length of stay outcome. Within the non-septic patients, the $p$-values are 0.67 and 0.77 for the complication and length of stay outcomes respectively.

\section{Design sensitivity}
\subsection{A restatement of the favorable setting}
For ease of reading we restate the favorable situation under which the calculations in \S 7 of the manuscript proceed. We imagine that $\zeta_i^{(\lambda_0)}$ is generated as \begin{align}\label{eq:gen2}
\zeta_i^{(\lambda_0)} &= \epsilon_i + S_i(\lambda - \lambda_0),
\end{align} 
where $\epsilon_i = (Z_{i1}-Z_{i2})\{(Y_{i1}-Y_{i2}) - \lambda(D_{i1} - D_{i2})\}$ are the adjusted encouraged-minus-non encouraged differences in responses, and $S_i = (Z_{i1}-Z_{i2})(D_{i1}-D_{i2})$ are the encouraged-minus-non encouraged differences in the treatment received, reflecting the strength of the instrument. Note that this generative model does \textit{not} imply the proportional dose model, such that the individual-level effects are allowed to be heterogeneous.

We assume that $\epsilon_i$ are $iid$ from a symmetric distribution with mean zero and finite variance $\sigma^2$. The treatments received are assumed binary. We assume that there are no defiers, that the exclusion restriction holds, and that individuals $ij$ are assigned status as compliers, never-takers and always-takers independently with probability $p_C$, $p_N$, and $p_A$ respectively. This results in $\P(S_i=1) = p_C + p_Ap_N$, $\P(S_i=-1) = p_Ap_N$ and $\P(S_i=0) = 1-p_C-2p_Ap_N$. The true treatment effect among compliers is $\lambda$, while $\lambda_0$ is its value under the null.

\subsection{A formula for design sensitivity}

In a sensitivity analysis, bias dominates variance in large samples. Under mild regularity conditions there is a number $\tilde{\Gamma}$, the {design sensitivity}, such that the power of a sensitivity analysis tends to 1 if $\Gamma < \tilde{\Gamma}$ and tends to 0 if $\Gamma > \tilde{\Gamma}$ as $n\rightarrow \infty$ in the favorable situation of no bias in treatment assignment \citep{Rosenbaum:2004b}. Larger values for $\tilde{\Gamma}$ indicate reduced sensitivity of inferences to unmeasured confounding in large samples. 

For any value of $\Gamma$, our procedure employs the random variable $\bar{L}_\Gamma = n^{-1}\sum_{i=1}^n\{\zeta_i^{(\lambda_0)} - (\Gamma-1)/(1+\Gamma)|\zeta_i^{(\lambda_0)}|\}$. Under the favorable situation being considered this random variable has expectation $E(\zeta_i^{(\lambda_0)}) - (\Gamma-1)/(1+\Gamma)E|\zeta_i^{(\lambda_0)}|$, where $E(\zeta_i^{(\lambda_0)}) = p_C(\lambda-\lambda_0)$ and 
$E|\zeta_i^{(\lambda_0)}| = (p_C + 2p_Ap_N)E|\epsilon_i + (\lambda-\lambda_0)| + (1-p_C - 2p_Ap_N)E|\epsilon_i|$.
\setcounter{proposition}{6}
\begin{proposition}\label{prop:designsens}
Suppose that $\zeta_i^{(\lambda_0)}$ are drawn $iid$ from (\ref{eq:gen2}), and that $\epsilon_i$ are drawn $iid$ from a symmetric, mean zero distribution with finite variance. Then, the design sensitivity is\begin{align}\label{eq:designsens}
\tilde{\Gamma}&=  \frac{E|\zeta_i^{(\lambda_0)}| + E(\zeta_i^{(\lambda_0)})}{E|\zeta_i^{(\lambda_0)}| -E(\zeta_i^{(\lambda_0)})}.
\end{align}
\end{proposition}
\begin{proof}
Under $iid$ draws from a distribution with finite variance the strong law of large numbers applies to $n\times \text{se}(\bar{L}_\Gamma)^2$. Therefore, Proposition 2 of \citet{Rosenbaum:2013} holds, and Corollary 1 of \citet{Rosenbaum:2013} yields the formula for the design sensitivity.
\end{proof}
\begin{remark}
An analogous proof shows that $\varphi_{\bQ}^{(\lambda_0)}$ also has the design sensitivity given in (\ref{eq:designsens}). In the limit, sensitivity to hidden bias is determined by the extent to which a test statistic can control the worst-case bias at $\Gamma$, and discrepancies in variance become irrelevant. So while exploiting effect modification can provide improvements in power at $\Gamma=1$, there is a sense in which improvements provided by exploiting effect modification in a sensitivity analysis are confined to small and moderate sample sizes.
\end{remark}


The formula for design sensitivity in (\ref{eq:designsens}) demonstrates the dependence of a sensitivity analysis's limiting power on various components of the data generating process. While the magnitude of the effect relative to its postulated value matters, so too does the proportion of compliers relative to always-takers and never-takers. The variance of $\epsilon_i$ (the degree of within-pair heterogeneity) plays an important role through its influence on $E|\epsilon_i|$ and $E|\epsilon_i + \lambda - \lambda_0|$. These results help explain the phenomenon observed in our data, where larger treatment effect estimates in the septic subgroup were less robust to hidden bias than the smaller estimates in the non-septic subgroup.

Consider the length of stay outcome variable, and suppose that the estimated effect ratios are actually the true values of $\lambda$. Suppose further that the standard deviation of $\zeta_i^{(\lambda_S)}$ and $\zeta_i^{(\lambda_{NS})}$ in our sample reflect the true standard deviations for the distribution of $\epsilon_i$ in the septic and non-septic subgroups. Table \ref{tab:DS} gives the design sensitivities from (\ref{eq:designsens}) using these values for $\lambda$ and $\sigma$ under both Normal and Laplace distributions while varying the proportion of compliers.  In keeping with \citet{Small:2008}, we observed that the weaker the instrument, the lower the design sensitivity becomes. Furthermore, we see that for a given level of compliance and regardless of distribution, the parameter settings motivated by the non-septic subgroup, with a \textit{lower} effect ratio but lower heterogeneity, outperforms the septic subgroup in terms of design sensitivity. This illustrates that in the limit, a smaller treatment effect may prove more robust to hidden bias when the data exhibit lower heterogeneity.

\begin{table}[h]
\begin{center}
\caption{\label{tab:DS} Design sensitivity calculations with parameters inspired by septic and non-septic length of stay effect estimates, varing the percentage of compliers and the distribution for $\epsilon_i$ in (\ref{eq:gen2}). Calculations assume there are no defiers, and that noncompliers are equally likely to be always-takers and never-takers. The columns in bold use the estimated compliance percentage from our data.}
\begin{tabular}{c c c c c c c c   c c c c c c}
\toprule
&\multicolumn{6}{c}{Normal} &\multicolumn{6}{c}{Laplace}\\
\cmidrule(lr){2-7}
\cmidrule(lr){8-13}
\multicolumn{1}{r}{Compliance} & 100\% & 75\% & \textbf{58\%}& 50\% & 25\% & 10\%& 100\% & 75\% & \textbf{58\%}&50\%  & 25\% & 10\%\\
\midrule
Septic & \multirow{2}{*}{1.97} & \multirow{2}{*}{1.65} & \multirow{2}{*}{\textbf{1.47}} &\multirow{2}{*}{1.39} &\multirow{2}{*}{1.18}&\multirow{2}{*}{1.07}&\multirow{2}{*}{2.11} &\multirow{2}{*}{1.75}& \multirow{2}{*}{\textbf{1.54}} & \multirow{2}{*}{1.45} & \multirow{2}{*}{1.20} & \multirow{2}{*}{1.08}\\

$\lambda = 6.8;\;\; \sigma = 25.3$\\
Non-septic & \multirow{2}{*}{3.19} & \multirow{2}{*}{2.33} &\multirow{2}{*}{\textbf{1.91}} & \multirow{2}{*}{1.74} &\multirow{2}{*}{1.32}&\multirow{2}{*}{1.12}&\multirow{2}{*}{3.50} &\multirow{2}{*}{2.51}& \multirow{2}{*}{\textbf{2.02}}& \multirow{2}{*}{1.83} & \multirow{2}{*}{1.35} & \multirow{2}{*}{1.13}\\
$\lambda = 4.1;\;\; \sigma = 8.9$\\
\bottomrule
\end{tabular}
\end{center}
\end{table}

\bibliographystyle{apalike}
\bibliography{er_bib}

\end{document}